\documentclass[11pt]{article}

\usepackage[ linesnumbered, lined,boxed,commentsnumbered, ruled]{algorithm2e}
\SetAlFnt{\small}

\usepackage{amsmath, amsthm, amsfonts, amssymb, float}
\usepackage{algpseudocode, cite, graphicx, times, thmtools, thm-restate}
\usepackage[full]{complexity}
\usepackage[margin=1.5in]{geometry}

\usepackage[textwidth=1.25in, disable]{todonotes}

\usepackage{tocloft}
\usepackage{stmaryrd}

\usepackage[dvipsnames]{xcolor}

\SetCommentSty{mycommfont}

\usepackage{enumitem}
\setlist{listparindent=\parindent,parsep=0pt,itemsep=1em}
\setlist[itemize]{label=$-$,noitemsep}
\setlist[enumerate]{itemsep=1mm}
\setlist[description]{leftmargin=\parindent}

\usepackage[full]{complexity}
\usepackage{bbm}
\usepackage{mathpartir} 
\usepackage{tabularx}   
\usepackage{array}      
\colorlet{mylinkcolor}{violet}
\colorlet{mycitecolor}{violet} 
\colorlet{myurlcolor}{violet} 

\usepackage{tikz}

\usetikzlibrary{arrows.meta, decorations.pathmorphing, shapes.geometric}

\usepackage{hyperref}
\hypersetup{
	colorlinks=true,
	urlcolor=MidnightBlue,
	linkcolor=MidnightBlue,
	citecolor=MidnightBlue,
	unicode
}
\definecolor{betterYellow}{RGB}{255,255,0}
\usepackage[capitalise,nameinlink,noabbrev]{cleveref}
\crefname{equation}{}{}

\definecolor{mutedsalmon}{RGB}{210, 140, 120}

\definecolor{mutedgreenlight}{RGB}{160, 195, 160}

\definecolor{mutedgreen}{RGB}{100, 150, 100}

\declaretheorem[name=Theorem, numberwithin=section]{theorem}
\declaretheorem[name=Corollary, sibling=theorem]{corollary}
\declaretheorem[name=Proposition, sibling=theorem]{proposition}
\declaretheorem[name=Lemma, sibling=theorem]{lemma}

\declaretheorem[name=Problem, sibling=theorem]{problem}

\theoremstyle{definition}
\declaretheorem[name=Definition, sibling=theorem]{definition}

\theoremstyle{remark}

\declaretheorem[name=Claim, sibling=theorem]{claim}

\usepackage{stmaryrd}

\newcommand{\set}[1]{\{ #1 \}}

\newcommand{\coloneq}{:=}

\newcommand{\NN}{\ensuremath{\mathbb{N}}}

\newcommand{\bits}{\ensuremath{\set{0,1}}}
\renewcommand{\epsilon}{\varepsilon}

\newcommand{\dtR}{R^{dt}}
\newcommand{\dtS}{S^{dt}}
\newcommand{\Resolution}{\ensuremath{\mathrm{Resolution}}}
\newcommand{\Reduction}{\ensuremath{\mathrm{Reduction}}}
\newcommand{\ImpProof}{\ensuremath{\mathrm{ImpProof}}}
\newcommand{\Proof}{\ensuremath{\mathrm{Proof}}}
\newcommand{\unsat}{\ensuremath{\mathrm{UNSAT}}}
\newcommand{\search}{\ensuremath{\mathrm{Search}}}

\newcommand{\Eval}{\ensuremath{\mathrm{Eval}}}

\newcommand{\isCkt}{\ensuremath{\mathrm{isCkt}}}

\newcommand{\WrongProof}{\ensuremath{\mathrm{WrongProof}}}
\newcommand{\refl}{\ensuremath{\mathrm{Refl}}}

\newcommand{\newprob}[2]{\newcommand{#1}{{\text{\upshape\scshape #2}}\xspace}}
\newprob{\sod}{SoD}
\newprob{\sodlong}{Sink-of-Dag}
\newprob{\iter}{Iter}
\newprob{\iterlong}{Iteration}

\renewcommand{\include}{\input}

\newclass{\VPV}{VPV}
\newclass{\V}{V}
\newclass{\TV}{TV}
\newclass{\T}{T}
\newclass{\PV}{PV}
\newlang{\isod}{iSoD}
\newclass{\Gpf}{G}
\newlang{\Rfn}{Rfn}
\newclass{\EF}{EF}
\newlang{\Iter}{Iter}
\newcommand{\Res}{\ensuremath{\mathrm{Res}}}

\newcommand\calcturnstile[1]{%
  \ifcase#1\let\tsstyle\textstyle\let\tsvstyle\textstyle\let\tsfont\textfont\or
    \let\tsstyle\textstyle\let\tsvstyle\textstyle\let\tsfont\textfont\or
    \let\tsstyle\scriptstyle\let\tsvstyle\textstyle\let\tsfont\scriptfont\or
    \let\tsstyle\scriptscriptstyle\let\tsvstyle\scriptstyle\let\tsfont\scriptscriptfont\fi}

\newcommand\turnstile[3][2]{%
  \mathrel{\calcturnstile{#1}%
  \setbox0=\hbox{$\tsvstyle\vdash$}%
  \setbox2=\hbox{$\vcenter{\copy0}$}%
  \hbox{\vrule\vphantom{$\tsvstyle\vdash$}}%
  \raise\dimexpr\ht0-\ht2\relax\hbox{$\vcenter{\offinterlineskip
   \ialign{\hfil\kern1pt$\tsstyle##\vphantom{by}$\kern1pt\hfil\cr
           \relax\if\relax\detokenize{#2}\relax\hphantom{\vdash}\kern-2pt\else#2\fi\cr
           \noalign{\kern1pt\hrule\kern1pt}%
           \relax\if\relax\detokenize{#2}\relax\hphantom{\vdash}\kern-2pt\else#3\fi\cr}}$}%
}}

\title{Implicit Proofs}
\author{Noah Fleming, Stefan Grosser, Toniann Pitassi, Robert Robere}

\begin{document}
\begin{center}
{\huge Provable Reductions in $\TFNP$}
\\[9mm] \large

\setlength\tabcolsep{1em}
\begin{tabular}{cccc}
Noah Fleming&
Stefan Grosser&
Toniann Pitassi&
Robert Robere\\[-1mm]
\small\slshape Lund \& Columbia &
\small\slshape McGill  &
\small\slshape Columbia &
\small\slshape McGill 
\end{tabular}

\vspace{9mm}

\large
\today

\vspace{5mm}

\normalsize
\end{center}

\begin{abstract}
    We introduce a new family of propositional proof systems, denoted $\langle EF, R \rangle$, for an arbitrary $\TFNP$ search problem $R$.
    Informally, a refutation of a CNF formula $F$ in $\langle EF, R \rangle$ is given by a polynomial-time mapping reduction from the \emph{false-clause search problem} $\search_F$ to $R$, combined with an Extended Frege proof that the reduction is correct.
    These new systems are naturally motivated in two ways:
    \begin{enumerate}
        \item They are the propositional translations of \emph{witnessing theorems} in bounded arithmetic, by which proofs of $\forall \Sigma^b_1$ formulas $\phi$ in a theory $T$ imply algorithms solving the search problem for $\phi$ in a $\TFNP$ subclass corresponding to $T$ \cite{buss_book, krajicek1990quantified, BussK94, krajivcek2007np, BeckmannB09}.
        \item They form a white-box analogue of the recent characterizations of proof systems using \emph{decision tree reductions} to black-box $\TFNP$ problems \cite{GoosH0MPRT22, Beame1994, DavisR23, LiPR24, FlemingG00RS026, FlemingIM25}. 
    \end{enumerate}
    We consider the proof system $\langle EF, \iter \rangle$, where $\iter$ is the complete problem for the classical $\TFNP$ subclass $\PLS$.
    We prove that $\langle EF, \iter \rangle$ is polynomially equivalent to the quantified boolean sequent calculus $G_1$, and also to the implicit Resolution proof system $[EF, \Resolution]$ introduced by Krajíček \cite{K04}. 
    Hence $G_1$ and $[EF, \Resolution]$ are polynomially equivalent, which is the first natural characterization of an implicit proof system by a classical propositional proof system beyond the work of Wang \cite{Wang13g}.
    We further observe our characterization can be extended to capture $G_i$ via the \emph{game induction principles} of \cite{skelley2007provably}.

    We also calibrate the strength of $\langle EF, R \rangle$ for general $\TFNP$ relations $R$.
    We observe that if Extended Frege can prove that a search problem $R$ is in $\FP$, then $\langle EF, R \rangle$ is polynomially equivalent to $EF$.
    This contrasts to our above result, which shows that Extended-Frege provable reductions to $\iter$, a problem widely believed \emph{not} to be in $\FP$, yields a proof system ($G_1$) that is believed to be stronger than Extended Frege.

    Finally, and somewhat paradoxically, we show that for \emph{any} proof system $P$ which is sufficiently strong, there is a polynomial-time computable search problem $R_P \in \FP$ such that $\langle EF, R_P \rangle$ is polynomially equivalent to $P$.
    Letting $P = [EF, \Resolution]$ and combining our two results shows that $\langle EF, \iter \rangle$ is polynomially equivalent to $\langle EF, R_{[EF, \Resolution]} \rangle$. 
    Hence, despite the widely-believed conjecture that $\iter \not \in \FP$, the problems which $EF$-provably reduce to $\iter$ are exactly the problems which $EF$-provably reduce to a fixed polynomial-time computable set.

\end{abstract}
\setcounter{tocdepth}{2} 
\newpage
\tableofcontents

\newpage

\section{Introduction}

A central goal in propositional proof complexity is to understand the relative power of proof systems and, ultimately, to separate $\NP$ from $\coNP$
by proving superpolynomial lower bounds on proof lengths. Despite decades of progress on concrete systems such as Resolution, bounded-depth Frege, and Cutting Planes, our understanding remains fragmented: the proof systems for which we have unconditional lower bounds captures a fairly narrow class of relatively weak proof systems,
and extending these lower bounds  even for $AC^0_p$-Frege systems has been a longstanding open problem.

More fundamentally, even if we managed to develop
techniques for proving lower bounds for stronger proof systems (such as Frege and Extended Frege ($EF$) systems), it is conjectured that there is no single optimal proof system. 
On the other hand, despite much effort, there is still no evidence of a family of concrete tautologies (beyond reflection principles) that are believed to be hard for $EF$ but easy for a stronger proof system. 

In a beautiful paper  \cite{K04} (see also \cite{krajicek2001implicit}), Krajíček proposed a method to create stronger propositional proof systems from weaker ones by encoding proofs in a weaker proof system implicitly as succinct circuits, together with short certificates of their correctness.
Concretely, given a Cook-Reckhow propositional proof system $P$,
the \emph{implicitization operator} applied to $P$ gives a new {\it implicit} proof system, denoted $iP$ (cf.~\cref{sec:prelims}).
An $iP$ proof is a pair $(\Pi,C)$ such that $C$ is a succinct representation of a $P$ proof, and $\Pi$ is a $P$-proof that $C$ encodes a valid $P$-proof.
By analogy to jump operators in classical logic, Krajicek conjectured that implicit proofs may give a way to create a {\it strictly} stronger
proof system from a weak one.
Iterating this construction generates a natural hierarchy of increasingly powerful systems which closely mirror the reflection or consistency hierarchies in bounded arithmetic. 
Therefore, implicit proofs can be seen as a concrete framework for studying how much additional strength we get by augmenting a proof system with the ability to reason about its own proofs.

Viewed from the lens of complexity theory,
the idea behind implicit proofs is analogous to the definitions of  succinct complexity classes.
Formally, if $L$ is a language, then we can define the \emph{succinct} version of $L$ as follows: given a polynomial-size circuit $C$ as input, decide if $tt(C) \in L$, where $tt(C)$ is the truth table of $C$ obtained by evaluating $C$ on all possible inputs.
For example, it is well-known that given any $\NP$-complete problem, its succinct version is $\NEXP$-complete, and similarly for any problem that is complete for $\P$ under projections, its succinct version is $\EXP$-complete.
In both cases, compression is used to ``jump'' from a lower complexity class or weaker proof system
to a stronger complexity class or proof system. 
Notably, compressed/succinct versions of complexity classes
are known to be provably stronger than the base class, via time hierarchy theorems, whereas in
contrast, it is a major open problem whether the succinct versions of proof systems are generally stronger
than the original proof systems.

Despite the intrinsic appeal and motivation behind the implicit proof framework, very little is currently known about their strength, and how they relate to more standard proof systems. 
At the high end (i.e.~for strong proof systems that can simulate $EF$), Pudl\'ak \cite{Pudlak2020} recently conjectured a striking connection between implicit Extended Frege proofs and
the consistency extension of Buss's theory $S^1_2$ of bounded arithmetic: iterated implicit operations on EF, $i^kEF$, capture the strength of $S^1_2 + \mathrm{Con}(S^1_2)$. Specifically, Pudl\'ak suggests that the provable $\forall\Sigma^b_1$-formulas in $S^1_2 + \mathrm{Con}(S^1_2)$ are exactly those which have propositional proofs in $i^kEF$, for some $k\geq 0$. 
Other work has studied the strength of $iEF$ as a tool for formalizing results in complexity theory. Khaniki \cite{Khaniki2024Jump} showed that $iEF$ is strong enough to formalize the soundness of the sum-check protocol. This was further used
\cite{Arteche2024} to show that if $iEF$ proves strong derandomization assumptions, then $iEF$ lower bounds
imply $\#\P \not\subset \FP/\text{poly}$. 
Given these results and the jump operator origins of Kraj\'i\v cek, it is clearly important to understand the strength of the implicit operator and what
implicit proofs can formalize.

At the low end -- for proof systems seemingly weaker than Extended Frege,
much less is known or conjectured.
Wang \cite{Wang13g} proved that implicit tree-like Resolution is polynomial-equivalent to $EF$\footnote{Wang refers to implicit tree-like Resolution as ``implicit Resolution''.}
However, analogous characterizations have not yet been shown for other standard propositional proof systems.
For example, for well-studied propositional proof systems  such as Resolution or Frege, what can we say about their implicit versions, $i\Res$ and $i\mathrm{Frege}$?
Can they be characterized by a natural proof system? How strong are they?

\subsection{Main Results}


We obtain characterizations of implicit propositional proof systems at the low end, for a variety of natural proof systems.
As a concrete example, we prove the following theorem, and will use it as a running example.
Recall that $G$ is a proof system for reasoning about \emph{quantified} boolean formulas. 
Formally, $G$ is a sequent calculus where each individual formula is a quantified boolean formula, and $G_i$ is the fragment of $G$ in which cuts are restricted to formulas of the form $\exists \vec x_1 \forall \vec x_2 \cdots Q_i \vec x_i F(\vec x_1, \dots, \vec x_i, \vec y)$, where $Q_i$ is $\exists$ if $i$ is odd and $\forall$ if $i$ is even.
For proving ordinary propositional statements, $G_1$ can be viewed as an extension of Extended Frege where proofs can reason with the more general family of quantified Boolean formulas of the form $\exists {\overline y} F({\overline x}, {\overline y})$, where $F$ is a propositional formula.  
We also note that Krajíček \cite{krajicek2001implicit} observed  that Implicit Resolution is equivalent to the (seemingly stronger) system $[EF, \Resolution]$, in which the proof of correctness of the succinct Resolution proof $C$ is in the system $EF$, rather than Resolution.
Our first main theorem is the following characterization of Implicit Resolution.

\begin{theorem}[cf.~\cref{thm:Res-G1}]
\label{thm:Res-G1-intro}
Implicit Resolution, and equivalently $[EF, \Resolution]$, is polynomially-equivalent to $G_1$.
\end{theorem}


This generalizes Wang's result for tree-like Resolution to general (dag-like) Resolution.



\paragraph{Methodology.} The proof of our theorem crucially relies on new tools that we develop using the theory of \emph{total functions in  $\NP$} ($\TFNP$), which has had an expanding role on propositional proof complexity in recent years \cite{GoosH0MPRT22, GoosMR025, LiPR24, BussFI23, GoosKRS19, DavisR23,HubacekKT24, FlemingIM25, FlemingG00RS026, FlemingGIM26, FlemingGPR24, kamath2019}.
In particular, we take inspiration from a recent line of work (e.g., \cite{GoosKRS19,BussFI23}) showing that black-box (low-depth decision tree) reductions to black-box $\TFNP$ problem are equivalent to short proofs in a corresponding propositional proof system. 
Formally speaking, if $F(x_1, \dots, x_n) = C_1 \land \cdots \land C_m$ is an unsatisfiable CNF formula, we can associate with $F$ the following \emph{false-clause search problem} $\search_F$: given an assignment $x$ to the variables of $F$, output the index $i$ of a false clause $C_i(x) = 0$.
It has long been known that efficient decision trees for $\search_F$ correspond directly to tree-like Resolution refutations for $F$.
Indeed, the \emph{totality} of the search problem $\search_F$ (for every input $x$, there is a valid output $i$ such that $C_i(x) = 0$) is equivalent to the fact that $F$ is an unsatisfiable formula.
Recent work has leveraged this observation to show that other classical propositional proof systems can be also captured, now  by giving efficient decision-tree \emph{reductions} from $\search_F$ to other search problems which are known to be total.
For instance, the following theorem of Kamath shows how to capture low-width Resolution refutations in this way via reductions to the \emph{black-box $\iterlong$ problem} (see \cref{sec:tfnp-prelims} for background on $\TFNP$).

\begin{theorem}[cf.~\cite{kamath2019}]\label{thm:kamath-res-iter}
    Let $F$ be an unsatisfiable CNF formula.
    If there is a width-$w$ Resolution refutation of $F$, then there is a depth $O(w)$ reduction from $\search_F$ to $\iter^{dt}$.
    Conversely, if there is a depth-$w$ reduction from $\search_F$ to $\iter^{dt}$, then there is a width $O(w)$ Resolution refutation of $F$. 
\end{theorem}

Our new characterizations of implicit proof systems proceed by considering a natural extension of the above idea.
Instead of considering reductions from $\search_F$ to other \emph{black-box} $\TFNP$ problems, we will consider reductions from $\search_F$ to classical \emph{white-box} $\TFNP$ problems using normal polynomial-time algorithms.
Unlike a decision tree reduction, we cannot explicitly verify that the reduction is correct, so we augment the proof system by adding an Extended Frege proof of correctness.
Formally, for \emph{any} $\TFNP$ relation $R \subseteq \bits^* \times \bits^*$, we introduce a new propositional proof system $\langle EF, R \rangle$ defined as follows.

\begin{definition}[Informal, see \cref{def:ef-mapping-r}]
   Let $F(x_1, \dots, x_n) = C_1 \land C_2 \land \cdots \land C_m$ be an unsatisfiable CNF formula, and let $R \in \TFNP$ be any total $\NP$ search problem.
   An \emph{$\langle EF, R \rangle$ refutation} of $F$ is a tuple $(C, D, \Pi)$, where $C$ and $D$ are polynomial-size circuits encoding a mapping reduction from $\search_F$ to $R$, and $\Pi$ is an Extended Frege proof that $C$ and $D$ are a correct mapping reduction.
   We call $(C, D, \Pi)$ an \emph{$EF$-provable mapping reduction} from $\search_F$ to $R$.
\end{definition}

We argue that this definition is the correct ``white-box analogue'' of the decision-tree reductions used in the black-box setting.
(Indeed, the previous definition can be viewed as an implicit version of the proof systems that can be characterized by black-box decision tree reductions, as studied by \cite{BussFI23}.)
In \cref{sec:G1}, we prove \cref{thm:Res-G1-intro} in two steps, using $\langle EF, R \rangle$ systems as a crucial intermediate step.
Below, the problem $\iter$ is the canonical $\iterlong$ problem, which is known to be complete for the $\TFNP$ class $\PLS$.

\begin{figure}[ht]
    \centering
        \begin{tikzpicture}[scale=1.1]
            \tikzset{inner sep=0,outer sep=3}
            
            \tikzstyle{a}=[inner sep=6pt, inner ysep=6pt,outer sep=0.5pt,
            draw=black!40!white, fill=Cerulean!10!white, very thick, rounded corners=6pt, align=center]
            \tikzstyle{b}=[inner sep=4pt, inner ysep=4pt,outer sep=0.5pt,
            draw=black!20!white, fill=Cerulean!10!white, thick, align=center]
            \large
                \node[a] (G1) at (-0.4,0) {$G_1$};
                \node[a] (Mapping) at (2.8,0) {$\langle EF, \Iter \rangle$};
                \node at (0.8,-0.1) {$\equiv_p$};
                \node[a] (Mapping) at (7.2,0) {$[EF, \mathrm{Resolution}]$};
                \node at (4.8,-0.1) {$\equiv_p$};
    
                \node at (0.9,-0.5) {\small \cref{claim:g1_to_mapping}};
                \node at (4.7,-0.5) {\small \cref{claim:res_eq_pls}};
            \end{tikzpicture}
\end{figure}

A fruitful way of interpreting the new proof system $\langle EF, R \rangle$ is as a propositional translation of the famous and central \emph{witnessing theorems} in bounded arithmetic\footnote{In the introduction we discuss the usual, single-sorted theories in bounded arithmetic. However, we warn the reader that, later in the paper, we have found it more convenient to do our formalizations in the two-sorted setting à la Cook-Nguyen \cite{CN2010}.}.
Informally, a witnessing theorem shows that for a bounded arithmetic theory $T$, a $T$-proof of a tautology of the form $\forall x \exists y \leq t: \phi(x, y)$ (these are called $\forall \Sigma^b_1$-statements) implies the existence of an algorithm from the corresponding $\TFNP$ class which on input $x$, outputs a $y$ such that $\phi(x,y)$ holds.
In this paper, we exploit a witnessing theorem due to Buss and Krajíček \cite{BussK94}, who showed that the $\forall \Sigma^b_1$ statements provable in the bounded arithmetic theory $T^1_2$ are witnessed by reductions to the $\TFNP$ class $\PLS$. 
A follow-up work of Beckmann and Buss \cite{BeckmannB09} generalized this theorem to higher $T^i_2$ classes, and showed that these theorems can themselves be formalized in $S^1_2$.
The first step above (see cf.~\cref{claim:g1_to_mapping}) follows by known theorems, and is essentially obtained by combining a propositional translation of the provable witnessing due to \cite{BeckmannB09} of $T^1_2$ by $\PLS$ algorithms (and hence by reductions to $\iter$), and combining it with the known relationships between $G_1$ proofs and tautologies in $T^1_2$.

The second step (cf.~\cref{claim:res_eq_pls}) we view as one of our main contributions, and requires new ideas.
To prove this, we take direct inspiration from \cref{thm:kamath-res-iter} characterizing low-width Resolution by black-box reductions to $\iter$.
For the reader familiar with bounded arithmetic, one intuitive way think of \cref{claim:res_eq_pls} is by formalizing the proof of \cref{thm:kamath-res-iter} in the \emph{relativized} theory $S^1_2(\alpha)$, where the oracle $\alpha$ will code a Resolution proof or decision tree reduction, depending on the direction of the equivalence.
We can then substitute $\alpha$ for a circuit everywhere in the proof and take a propositional translation, using the fact that $S^1_2$ proofs propositionally translate to Extended Frege proofs \cite{buss_book}.

While the proof can be formalized in this way, we opt to work directly instead, as it is more illuminating.
In one direction, from an implicit Resolution refutation $(C, \Pi)$ of a formula $F$, we will directly create a mapping reduction from $\search_F$ to $\iter$ by mimicking the standard \emph{Prover-Delayer game} for Resolution \cite{Pudlak00}.
For each line in the implicit proof encoded by $C$, we create a node in our $\iter$ instance, and given an assignment $x$ which falsifies the line, we define the successor to be the clause used to derive that line which is also falsified.
In this way, solutions of the $\iter$ instance will correspond directly to false clauses of $F$, and we can use the Extended Frege proof $\Pi$ of the correctness of the encoding of $C$ to argue that this is a provably correct reduction from $\search_F$ to $\iter$.

The converse direction is proved similarly by following the converse direction of the proof of \cref{thm:kamath-res-iter}.
Now, from a provable reduction from $\search_F$ to $\iter$ we must construct an implicit Resolution proof $(C, \Pi)$.
From the reduction we explicitly construct this proof, but we note that a remarkable feature is that the proof we construct is \emph{highly} redundant: the Resolution proof we build will have size $2^{\poly(n)} \gg 2^n$, and much of it looks like re-wiring various redundant clauses to one-another.
We refer the reader to \cref{sec:G1} for more details.




\paragraph{Characterizing $G_i$, for $i>1$.}
Using the same proof skeleton as \cref{claim:g1_to_mapping} (using an $S^1_2$-provable witnessing theorem for $T^i_2$ \cite{BeckmannB09}), we are able to obtain similar characterizations of implicit depth-$d$ Frege systems. We use the \emph{game induction principles} $\mathsf{GI}_d$ of Skelley and Thapen \cite{skelley2007provably}, which are the black-box $\TFNP$ problems characterizing depth-$d$ Frege.

\begin{theorem}
    For all $d>0$, $\langle EF, \mathsf{GI}_d\rangle$ is polynomially equivalent to $G_d$.
\end{theorem}

We note that this does not directly imply that $G_d = [EF, \mathsf{GI}_d]$, without a generalization of \cref{claim:res_eq_pls}.

\begin{problem}
     Let $\mathrm{Frege}_d$ denote the fragment of Frege proofs in which each line has alternation depth $d$. Is $\langle EF, \mathsf{GI}_d\rangle \equiv_p [EF, \mathrm{Frege}_d]$?
\end{problem}

\noindent We leave this problem to follow-up work, due to its considerable technicality in comparison to \cref{claim:res_eq_pls}, as well as its independent interest.

These theorems can be viewed as supporting evidence for a main
message of this work:  {\it well-studied strong proof systems for propositional reasoning are compressible instances of a corresponding natural weak proof system.} 

\paragraph{Generalizing \cref{claim:res_eq_pls}.}

Our second main result (cf.~\cref{sec:generic}) shows that the connection proved in \cref{claim:res_eq_pls} holds generally.
Our work here is inspired by recent work of Buss, Fleming, and Impagliazzo~\cite{BussFI23}, who obtained a general characterization of black-box $\TFNP$ problems and propositional proofs. 
They show that for every (sufficiently uniform) black-box problem $R \in \TFNP^{dt}$, there is a corresponding propositional proof system $P$ such that a $\TFNP^{dt}$ problem $Q$ is decision-tree reducible to $R$ if and only if $P$ can prove that $Q$ is total.
In \cref{sec:generic} we prove a similar relationship between strong proof systems $P$ (those that can simulate $EF$) and provable mapping reductions.
In particular, we view the next theorem as the correct ``white-box'' analogue of the known connections between proof systems and black-box reductions \cite{BussFI23}.

\begin{theorem}[Informal, cf.~\cref{thm:generalEquiv}]
   Let $P \geq_p EF$ be any proof system which has polynomial-size proofs of its own reflection principle.
   Then $P$ is polynomially-equivalent to $\langle EF, \WrongProof_P \rangle$.
\end{theorem}

In the above theorem, $\WrongProof_P$ is the following total search problem: given a CNF formula $F$, an assignment $x$ to the variables of $F$, and a (proposed) $P$-refutation $\Pi$ of $F$, either output a false clause of $F$ under $x$ or find an error in the proof $\Pi$.
By applying the above theorem when $P = [EF, \Resolution]$, we therefore obtain
\[G_1 \equiv_p \langle EF, \iter \rangle \equiv_p [EF, \mathrm{Resolution}] \equiv_p \langle EF, \WrongProof_{[EF, \Resolution]}\rangle. \]

We believe the characterization $\langle EF, \iter \rangle \equiv_p \langle EF, \WrongProof_{[EF, \Resolution]} \rangle$ is quite surprising in its own right.
The $\iter$ problem is complete for $\PLS$, which is widely believed to be different from $\FP$.
However, for \emph{any} proof system $P$, $\WrongProof_P \in \FP$ (!), since given $F, x$ and $\Pi$, we can easily verify in polynomial time if $x$ falsifies a clause of $F$ or if $\Pi$ contains an incorrect proof step.
Hence, even though $\iter$ and $\WrongProof_{[EF, \Resolution]}$ likely have very different complexity as \emph{search problems} (with respect to unrestricted polynomial-time reductions), if we restrict ourselves to \emph{$EF$-provable} reductions they capture the same set of problems.

\paragraph{Provably Correct Algorithms.} As far as we are aware, the only known example in the literature of a polynomial time algorithm whose ``proof of correctness" is not in $\S^1_2$ is the AKS primality testing algorithm \cite{aks}. This follows by witnessing argument that if $\S^1_2$ proves ``AKS($p$) = 1 if and only if $p$ is prime'', then factoring is in $\FP$ \cite{CN2010}. The best upper bound for formalizing the correctness of AKS is the theory $\mathsf{VTC}^0_2$, which corresponds to quasipolynomial size $\TC^0$ reasoning, and is incomparable with the full Buss Hierarchy $S_2 = \bigcup_{i=1}^\infty S^i_2$ \cite{jalali2026feasibility}.

Taken from the lens of bounded arithmetic, \cref{thm:generalEquiv} gives a suite of polynomial time algorithms whose \emph{correctness} is arbitrarily hard to prove. Consider the equivalence $G_1~\equiv_p~\langle EF,\WrongProof_{G_1}\rangle$, with $\WrongProof_{G_1} \in \FP$. Notice that any polynomial time algorithm $f$ solving $\WrongProof_{G_1}$ cannot \emph{provably} do so in $\S^1_2$, unless $\T^1_2$ has the same $\Pi^b_1$-consequences as $\S^1_2$. This holds not only for $G_1$, but for any proof system $P \geq_p EF$ that proves it own reflection principle and corresponds to a bounded arithmetic theory $\mathcal{T}_Pe$. If $\mathcal{T}_P \supseteq \S^1_2$, then $\WrongProof_P$ is provably in $\FP$ if and only if $\mathcal{T}_P$ is $\Pi^b_1$-conservative over $\S^1_2$.

%
%
%
%
%

\subsection{Related Work}
In recent independent work of Pudl\'ak and Thapen \cite{thapen_pudlak}, they also prove that Implicit Resolution is equivalent to $G_1$. Their proof is related to our own, using the known connections between the theory $\T^1_2$, $G_1$, and $\PLS$. They additionally show, using different techniques based on cut elimination, a $p$-equivalence between $G_i$ and \emph{narrow} implicit $G_{i-1}$, denoted as $[EF, G_{i-1}]^m$ by Kraj\'i\v cek \cite{krajicek2001implicit}, which indicates that the lines of the exponentially large $G_{i-1}$ proof are all polynomial size. 

Where this work differs to \cite{thapen_pudlak} is in our heavy reliance on $\TFNP$ and our new definition of the $\langle EF, R\rangle$ proof systems. We believe that our techniques paint a simple and clear picture of the implicit proof landscape, and explain how the implicit resolution characterization of ours and Pudl\'ak-Thapen can be generalized to many other settings. Additionally, our white-box analogue of Buss, Fleming, and Impagliazzo \cite{BussFI23} is unique to this work.

\subsection{Open Problems}

This work raises several questions. First, can we give implicit characterizations of algebraic and semi-algebraic proof systems, such as Nullstellensatz (NS), Cutting Planes, and SOS? 
In particular is there a natural system characterizing $[EF,NS]$?
Similarly, the  characterization of implicit Resolution by $G_1$ gives rise to natural candidate hard problems for EF. Are there similar natural (e.g., combinatorial or algebraic) candidate hard tautologies for other well-studied propositional proof systems? 

Secondly, it is important to point out that our results connecting implicit black-box TFNP to provable white-box TFNP classes are sensitive to the particular complete problem for the class. In particular, the equivalence holds with respect to complete problems where the input is a succinct circuit representation of an exponentially large object.
Notably, for complete problems (such as Nash) that are not defined as succinct black-box problems, our characterization does not apply and more generally the connection between white and black box problems in this context is far from clear.

Lastly we note the importance of metamathematical upper bounds in our work.
 Our characterization of implicit propositional proof systems hinge on the fact that witnessing theorems can be formalized in $S^1_2$, together with known connections between bounded arithmetic proofs (of $\Sigma^b_1$ statements) and their propositional translations.
 Similarly, \cref{claim:g1_to_mapping} hinges on formalizing the Buss--Fleming--Impagliazzo characterization of {\it every} black-box TFNP problem in $S^1_2$.
 This adds to a growing body of work on the metamathematics of complexity theory that are based on the metamathematics of bounded arithmetic itself  
 (e.g. \cite{Jerabek04,Jerabek07,BussKKK20, razborov-equiv}).
 Here we mention a few other formalizability problems 
 that are likely to have significant implications: 
 (1) characterizing the minimal theory that can prove the PCP theorem; (2) prove (or disprove) the existence of short EF proofs of other higher order witnessing theorems; (3) determine the minimal theory necessary to prove correctness of nonconstructive or nonuniform polynomial-time algorithms.

\medskip

\subsection*{Organization}
\Cref{sec:prelims} contains preliminaries including definitions of standard proof systems, and implicit proofs.
In \Cref{sec:EF_reductions} we define our new TFNP-based implicit proof systems, $\langle EF,R \rangle$, where $R \in \TFNP$.
\Cref{sec:G1} proves \Cref{thm:Res-G1}, and in \Cref{sec:generic} we prove the generalization for arbitrary TFNP problems (\Cref{thm:generalEquiv}). Finally, formal definitions of the proof systems that we work with, bounded arithmetic theories, as well as witnessing theorems are supplied in \Cref{sec:appendix}.

\section{Preliminaries}
\label{sec:prelims}

\subsection{Proof Complexity}

In this section we outline the necessary background on relevant topics in proof complexity.
We refer the reader to Kraj\'{i}\v{c}ek's monograph for further details \cite{Krajicek2019}.

\begin{definition}
A \emph{propositional proof system} $P$ is given by a polynomial-time algorithm $V: \{0,1\}^* \times \{0,1\}^* \rightarrow \{0,1\}$, known as its verifier, such that for every CNF formula $F$, \[\exists\, \Pi\colon\, \,V(F,\Pi) = 1 \iff F \text{ is unsatisfiable}.\]
Given two propositional proof systems $P$, $Q$, with verifiers $V_P$ and $V_Q$, we say that $P$ \emph{polynomially-simulates} $Q$, written $P \geq_p Q$, if there is a polynomial-time computable function $f$ such that $V_Q(F, \Pi) \implies V_P(F, f(F, \Pi))$ for all $F$ and $\Pi$.
\end{definition}

A standard example is the \emph{Resolution} system.

\begin{definition}
    Let $C_1, \dots C_m, D$ be a collection of clauses.
    A \emph{Resolution} proof of $D$ from $C_1, \dots, C_m$ is a sequence of clauses \[\Pi = (D_1, D_2, \ldots, D_s)\] where $D_s = D$ and, for each $i \leq s$, either $D_i = C_j$ for some $j$, or $D_i$ is derived from earlier clauses in the sequence by one of the following two rules:
    \begin{itemize}
        \item {\emph{Resolution}} $C \lor x, D \lor \neg x \vdash C \lor D$
        \item {\emph{Weakening}} $C \vdash C \lor D$
    \end{itemize}
    The proof is a \emph{refutation} if $D = \bot$ (the empty clause).
\end{definition}

\paragraph{Frege and Extended Frege.} 
Frege systems are a broad class of propositional proofs systems which are polynomially equivalent to Gentzen's propositional calculus, PK, which we define formally in \autoref{def:pk} in the appendix . Briefly, a Frege system consists of a finite set of \emph{inference schemas} $A_1,\ldots, A_k \rightarrow B$ which is both sound and implicationally complete. 
Proofs in a Frege system are a sequence of formulas $\Pi=L_1,\ldots, L_t$, where each successive formula is either an axiom, or is derived from previous formulas by application of (substitution into) one of the rule schemas. 

Extended Frege strengthens Frege by allowing for the introduction of \emph{extension variables}.

\begin{definition}
    Let $\Pi=L_1,\ldots, L_t$ be a Frege-derivation using variables $x_1,\ldots, x_n$. The \emph{extension rule} allows one to add the line $L_{t+1} := e \leftrightarrow A$, where $A$ is any formula over the variables $x_1,\ldots, x_n$, and the \emph{extension variable} $e$ is not among these variables. 
\end{definition}
\noindent
Extended Frege ($EF$) is the propositional proof system of $PK$, with the addition of the extension rule. 

By allowing the lines in our Frege proof to include existential and universal quantifiers, we move from $EF$ to $G$, the quantified propositional calculus. The fragment of $G$ we will deal with the most is $G_1$, which restricts lines to contain only universal or existential quantifiers. More generally, the lines in the $G_d$ proof system may have $d$ alternations of quantifiers. These are defined formally in \autoref{def:G} in the appendix.

\subsection{Reflection Principles}
For a propositional proof system $P$ with verifier $V$, let $\{V_{n,m,s}\}$ be a family of polynomial-size circuits computing $V$ on each input length. Here, $n$ is the number of variables of the CNF formula $F$ being proven, $m$ is the number of clauses of $F$, and $s$ is the size of the proof $\Pi$ being verified.
A \emph{reflection principle} is a propositional formula which captures the soundness of a proof system, and will be central to our later characterizations.
\begin{definition}
\label{def:reflection}
    For a proof system $P$ with verifier $V$, its \emph{reflection principle} is the family of tautologies 
    \[\refl_{P,n,m,s}(F,\Pi,x) := \neg\mathrm{Proof}_{n,m,s}(F,\Pi) \vee \mathrm{SAT}_{n,m}(F,x)\]
    where $\mathrm{Proof}_{n,m,s}$ is a formula (using the definition of $V_{n,m,s}$) asserting that $\Pi$ is a $P$-proof that $F$ is a tautology, and $\mathrm{SAT}_{n,m}$ is the CNF formula claiming that $x$ is a satisfying assignment to $F$. This captures the soundness of the proof system $P$: it asserts that if $\Pi$ is a $P$-proof of $F$ then every assignment satisfies $F$. 
    If the parameters $n, m, s$ are clear from context, then we will suppress the subscripts and call the formula $\refl_P$.

    A $\Sigma^q_1$-formula is a quantified boolean formula of the form $\exists \vec y G(x, y)$, where $G$ is unquantified.
    The reflection principle for $\Sigma^q_1$-formulas is defined similarly by changing $\mathrm{SAT}(F,x)$ to $\mathrm{SAT}_{\Sigma^q_1}(F,x)$, asserting that $F$, a $\Sigma^q_1$-formula, is satisfied by assignment $x$. 
    The formula $\mathrm{SAT}_{\Sigma^q_1}(F,x)$ may itself be written as a $\Sigma^q_1$-formula in the natural way. 
    We will denote this version of reflection as $\Sigma^q_1$-$\refl_P$.
\end{definition}
\noindent 
We note that reflection principles are more easily described in the first order setting of bounded arithmetic---we do so in \Cref{appendix:BA}.

\subsection{Implicit Proof Systems}\label{sec:implicit-proof-systems}

As briefly mentioned in the Introduction, \cite{krajicek2001implicit, K04} introduced implicit proof systems as a way to create stronger proof systems from weaker ones by compression.
We now briefly review the definitions of these systems.

We start with some complexity theoretic  motivation behind Krajicek's definition, by viewing implicit proofs as succinct versions of existing proof systems, similar to the development of succinct versions of complexity classes.
Recall that for a string $x \in \{0,1\}^N$ 
(viewed as a Boolean function on $\log N$ input variables), the natural succinct encoding of $x$ is a boolean circuit $S$ such that $tt(S) = x$, where $tt( \cdot )$ outputs the truth table of the given circuit. 
For a language $L \subseteq \{0,1\}^*$, the succinct version of $L$ consists of all polynomial-size circuits $S$ such that $tt(S) \in L$.
This ``succinct operator'' compresses a language, and it is well-known that if $L$ is $\NP$-complete, then the succinct version is $\NEXP$-complete.

The same idea can be applied to proofs, giving rise to the implicit operator $i$. 
For a proof system $P$ and a $P$-proof $\pi \in \{0,1\}^*$, we can consider a succinct encoding of it by circuit $S$ such that $tt(S) = \pi$.
Of course, such a compressed proof may not be polynomial-time verifiable since we need to decompress the circuit first.
To circumvent this, an $[EF, P]$ proof consists of both the circuit $S$, along with an Extended Frege proof that $tt(S)$ is a $P$-proof.
We describe this concretely in \cref{sec:setup}.

We now give a more detailed definition of implicit proof systems, including a rough overview of the encodings of formulas.
(As Kraj\'i\v cek notes \cite{K04}, the actual encoding is not important as long as it satisfies some basic properties.)

We start by fixing a description of the implicit Resolution system $[EF, \Resolution]$.
If $C$ is a clause over variables $x_1, \dots, x_n$, we encode $C$ by two strings $(P^{(C)}, N^{(C)}) \in \bits^n \times \bits^n$, where $P^{(C)}_i = 1$ iff $x_i$ occurs in the clause and $N^{(C)}_i = 1$ iff $\neg x_i$ occurs in the clause.
In this way, all strings $(P, N) \in \bits^{2n}$ encode a clause.

Next, suppose we have a Resolution refutation of a CNF formula $F(x_1, \dots, x_n) = C_1 \land \cdots \land C_m$, defined by a sequence of clauses $D_1, D_2, \ldots, D_s$ where $D_s = \bot$.
We can encode this refutation as a string as follows.
Each clause $D_i$ in the proof is encoded by a string of the form $P^{(D_i)} \circ N^{(D_i)} \circ t^{(i)} \circ L^{(i)} \circ R^{(i)}$, where $\circ$ denotes string concatenation.
The components of the string are:
\begin{itemize}
    \item $P^{(D_i)}, N^{(D_i)} \in \bits^n$ together encode the literals in the clause $D_i$, as described above. If $i \leq m$ note that $D_i = C_i$.
    \item $t^{(i)} \in \bits^2$ is a short string encoding how $D_i$ was derived in the proof. 
    The tag only has meaning if $i > m$, and depending on its value, $L^{(i)}$ and $R^{(i)}$ will take on different meanings. 
    Namely: 
    \begin{itemize}
        \item If $t^{(i)} = 00$ then $D_i$ is derived by weakening, and the string $L^{(i)}$ will index into the clause $D_j$ with $j < i$ such that $D_i$ is obtained by weakening $j$.
        \item If $t^{(i)} = 01$ then $D_i$ is derived from earlier clauses by Resolution, and the string $L^{(i)}, R^{(i)}$ are the indices of the two clauses used to derive $D_i$.
        \item If $t^{(i)} = 10$, then $D_i$ is a weakening of a clause from $F$.
        \item If $t^{(i)} = 11$ then this line is disabled, and is not used by the proof.
    \end{itemize}
\end{itemize}
The encoding of the proof is obtained by concatenating the encoding of each line of the proof, as described above.

Now that we have described how to encode a Resolution proof by a string, we can talk about how to implicitly represent a proof.
A circuit \[C : \bits^s \rightarrow \bits^2 \times \bits^{2n} \times \bits^{2s}\] is said to \emph{implicitly encode a Resolution refutation} of $F$ if the truth-table of $C$ \[tt(C) := C(0^s) \circ C(0^{s-1}1) \circ \cdots \circ C(1^s)\]
implicitly encodes a Resolution refutation of $F$.
Note that the length of the refutation may be as long as $2^s$.

It is possible to formalize propositionally the notion of a correct implicit resolution proof in a CNF we denote $\ImpProof_\Res(n,m,s,C,F)$. We carefully do this in \cref{sec:setup} for the first-order setting---informally, $\ImpProof_\Res$ encodes that a circuit $C$ is a valid circuit, and that its truth table encodes the lines of a resolution proof, with each line being either an axiom of $F$, deduced by weakening, or deduced by a resolution rule.

\begin{definition}
    The proof system $[EF, \Resolution]$ is defined as follows.
    If $F(x_1, \dots, x_n) = C_1 \land \cdots \land C_m$ is an unsatisfiable CNF formula, then a refutation in $[EF, \Resolution]$ of $F$ is given by a pair $(C, \Pi)$, where $C$ is an implicit encoding of a Resolution refutation of $F$, and $\Pi$ is an Extended Frege proof that $C$ is correctly encoded.
\end{definition}

This definition may be generalized to $[EF, P]$, for any propositional proof system $P$. Now, $C$ encodes a succinct $P$-refutation of $F$, and $\Pi$ similarly is an $EF$ proof of correctness of $C$. We remark that for static proof systems like Nullstellensatz or Sum-of-Squares, it is not immediately clear how to similarly encode such proofs. However, this can be done through ad-hoc modifications of proof encodings, or more generally through succinctly encoding the computation history of the \emph{verifier} $V_P$.

\subsection{\texorpdfstring{$\TFNP$}{TFNP} Preliminaries}
\label{sec:tfnp-prelims}

In this section, we introduce some of the background on total $\NP$ search problems that we regularly use throughout the paper.
A recurring theme will be systematic connections between the classical \emph{white-box} theory of $\TFNP$, and the more recently studied \emph{black-box} theory of $\TFNP$ (e.g., \cite{
Megiddo1991,Papadimitriou1994,Goos2018,Beame1998,GoosHJMPRT22-collapses,BussFI23,GoosH0MPRT22,KomargodskiNY19}).
We start by recalling the definition of an $\NP$ search problem.

\begin{definition}
    A relation $R \subseteq \bits^* \times \bits^*$ is 
    \begin{itemize}
        \item \emph{Total}, if for all $x \in \bits^*$ there is a $y \in \bits^*$ such that $(x, y) \in R$. 
        \item \emph{Polynomially bounded}, if there is a polynomial $\ell(n)$ such that for all $(x, y) \in R$, $|y| \leq \ell(n)$.
    \end{itemize} 
    An \emph{$\NP$ search problem} is a polynomial-time computable, polynomially bounded relation $R \subseteq \bits^* \times \bits^*$.
    The class of all total $\NP$ search problems is denoted $\TFNP$.
\end{definition}

We think of $R \subseteq \bits^* \times \bits^*$ as a search problem in the usual sense: on receiving an input $x \in \bits^*$, the goal is to output a $y \in \bits^*$ such that $(x, y) \in R$.
We also often use the predicative notation $R(x, y)$ to denote the relation ``$(x, y) \in R$''.
For a $\TFNP$ relation $R$, we will reserve $\ell$ to denote an upper bound on the length of the certificate of $R$ as a function of $n$.

Since $R$ is a polynomial-time computable predicate, we can encode the execution of $R$ on an input by a polynomial-size circuit family $V_R$.
We record this next, as it will be useful later.

\begin{definition}\label{def:tfnp-cnf-encoding}
    Let $R \in \TFNP$ and let $n, \ell$ be positive integers.
    Define $V_{R, n, \ell}$ to be a polynomial-size circuit such that $V_{R, n, \ell}(x, y) = R(x, y)$.
    If the parameters $n$ and $\ell$ are clear from context, then we will omit them and write $V_R(x, y)$.
\end{definition}

As usual, there are notions of reductions between $\TFNP$ problems. We record the standard notion of reductions for now.

\begin{definition}\label{def:mapping-reduction}
    Let $R, S \in \TFNP$.
    A \emph{mapping reduction}\footnote{In this paper, all mapping reductions will be polynomial-time computable.} from $R$ to $S$ is given by two polynomial-time computable functions $f$ and $g$ such that for all $x, y \in \bits^*$, 
    \[S(f(x), y) \implies R(x, g(y)).\]
    If $R$ is mapping reducible to $S$ we write $R \leq_m S$.
\end{definition}

The above notion is natural as if $S \in \FP$ and $R \leq_m S$ then $R \in \FP$: given $x$, we run the polynomial-time algorithm $A$ solving $S$ on $f(x)$, receive an output $y$, and then output $g(y)$.
A canonical example of an interesting $\TFNP$ problem is the $\iterlong$ problem, which we define next.

\begin{definition}
    The $\iterlong$ problem, also denoted $\iter$, is defined as follows.
    For any positive integer $n$, 
    let $\ell(n) = n^{O(1)}$ be a  polynomially-bounded input parameter.
    The input to the problem is a size parameter $1^n$ and a circuit $S: \bits^n \times \bits^\ell \rightarrow \bits^\ell$ of size polynomial in $n$.
    Given $n$ and $S$, a string $y \in \bits^\ell$ is a valid solution if either
    \begin{itemize}
        \item $y = 0^\ell$, if $S(x, 0^\ell) = 0^\ell$, or
        \item $y \neq 0^\ell$, if $S(x, y) <_{lex} y$, or 
        \item $y \neq 0^\ell$, if $S(x, y) >_{lex} y$ and $S(x, S(x, y)) = S(x, y)$.
    \end{itemize}
    Above $<_{lex}$ denotes the usual lexicographic ordering on strings. 
    The class of all search problems mapping reducible to $\iterlong$ is denoted by $\PLS$.
\end{definition}

\paragraph{Black-Box $\TFNP$.}

In black-box $\TFNP$ we are provided the input by query access, as it is now thought of as an exponentially-long string, and we are interested in performing \emph{reductions} via black-box algorithms, where the primary complexity measure is the number of queries made by the reductions.
One of the main themes of this paper is the link between white-box and black-box $\TFNP$, and the central example of a total search problem in black-box $\TFNP$ is the \emph{false-clause search problem} associated with an unsatisfiable CNF formula $F(x_1, \dots, x_n) := C_1 \land \cdots \land C_m$.
For this reason, we will briefly review these connections, although we will not need the tools from this section in this paper beyond the definition of the false-clause search problem below.

The false-clause search problem is the following: given an assignment $x \in \bits^n$ to the input variables $x_1, \dots, x_n$, output the index $i$ of any falsified clause $C_i(x) = 0$.
\begin{definition}
	If $F(x_1, \dots, x_n) = C_1 \land \cdots \land C_m$ is an unsatisfiable CNF formula, the \emph{false-clause search problem} $\search_F$ associated with $F$ is the total relation \[\search_F := \set{(x, i) \in \{0,1\}^n \times [m]: C_i(x) = 0}.\]
\end{definition}


More generally, in the theory of black-box $\TFNP$ we are interested in \emph{black-box total search problems}.

\begin{definition}
    Let $n$ be a positive integer. 
    A \emph{black-box search problem} is any total relation \[\dtR_n \subseteq \set{0,1}^{n} \times O_n,\] where $O_n$ is a polynomial-sized set of feasible solutions.
    The class $\TFNP^{dt}$ contains all sequences $\dtR = \set{\dtR_n}_n$ of black-box search problems such that the following holds: there is a universal constant $c$ such that for each input length $n$, and each $o \in O_n$, there is a $O((\log n)^c)$-depth decision tree $T_o$ satisfying $T_o(x) = \dtR_n(x,o)$.
\end{definition}

Following standard convention, we will usually think of $n$ as a parameter and consider problems on $\poly(n)$ input bits.
If we have a sequence of unsatisfiable CNF formulas $F = \set{F_n}$, each with $\poly(n)$ variables and $\poly(\log n)$ width, then we abuse notation and write $\search_F := \set{\search_{F_n}}_n$ to mean the associated sequence of search problems. 
We note that any such sequence lies in $\TFNP^{dt}$, since each feasible solution of $F_n$ is a $\poly(\log n)$-width clause which can be easily verified by a $\poly(\log n)$-depth decision tree.


Next, we introduce a notion of reduction between black-box search problems, encoded by polynomial-size decision trees.
If $T$ is a decision tree then let $|T|$ be the number of leaves of $T$.
\begin{definition}\label{def:formulation}  
	Let $\dtS_m \subseteq \{0,1\}^m \times O_S, \dtR_n \subseteq \{0,1\}^n \times O_S$ be black-box search problems.
	A \emph{decision-tree reduction} from $\dtS_m$ to $\dtR_n$ is given by functions $f = \{f_i: \bits^m \rightarrow
    \{0,1\}\}_{i \in [n]}$ and $g = \{g_o: \bits^n \rightarrow O_S\}_{o \in O_R}$ 
    each computable by decision trees $T_{f_i}, T_{g_o}$, such that for every $x \in \{0,1\}^m$ and $o \in O_R$,
	\begin{equation*}\label{eqn:mapping-reduction}
		(T_f(x), o) \in \dtR \implies (x, T_{g_o}(x)) \in \dtS. 
	\end{equation*} 
    where $T_f(x)$ denotes the $n$-bit string obtained by evaluating $T_{f_1}(x),\ldots,T_{f_n}(x)$.
    The \emph{depth} of the reduction is the maximum depth of any of the decision trees $\set{T_{f_i}}_i \cup \set{T_{g_o}}_o$.
    The \emph{complexity} of the reduction to be: \[\sum_{i=1}^n |T_{f_i}| + \sum_{o \in O_S} |T_{g_o}|.\] 
    
	We extend this definition to sequences as follows. 
	If $\mathsf{R}^{dt} = \{\dtR_n\}$ is a sequence of search problems, then $R^{dt}(\dtS_m)$ is the minimum complexity of a reduction from $\dtS_m$ to any problem in $R^{dt}$. 
	If $S^{dt} = \{\dtS_m\}$ is also a sequence of search problems, then we write $S^{dt} \leq^{dt}_m \mathsf{R}^{dt}$ if $\mathsf{R}^{dt}(\dtS_m)$ is polynomially-bounded in $m$.
\end{definition}

A few remarks on this definition are needed.
First, the definition above slightly differs from the usual definition of a decision-tree reduction, as seen in e.g.~\cite{GoosH0MPRT22}.
The usual notion defines the complexity of a reduction to be the maximum depth of any decision tree, plus the logarithm of the size of all decision trees.
This notion is useful as it captures the query complexity of reductions in $\TFNP$.
For us, the above notion, capturing \emph{just} size, will be more useful as it captures the lengths of proofs in a natural way.

Second, when actually defining reductions it will be more convenient for us to think of the decision trees $f$ as outputting more general objects, like elements in some range $[n]$, rather than individual bits or restrictions.
We will abuse notation and do this liberally, with the understood convention that whenever we output or query an element $r \in [n]$, this would be formalized by replacing $r$ with a $\log n$-length bit string encoding $r$.


To continue our motivating example, just as in the classical setting, the $\iterlong$ search problem has a natural black-box counterpart denoted $\iterlong^{dt}$. 
\begin{definition}
Let $N = 2^k$ be a positive integer. 
The problem $\iter_N^{dt}$ is defined as follows.
As input, for each element $u \in [N]$, we are given a \emph{successor} $s_u \in [N]$, interpreted as naming the successor node of $u$, and which we can think of as being encoded by a string of $k$ bits.
The goal of the search problem is to output any of the following
\begin{enumerate}
\item $1$ if $s_1 = 1$ \hfill \emph{(inactive distinguished source)}
\item $u \neq 1$, if $s_u < u$, \hfill \emph{(decreasing successor)}
\item $u$, if $u$ is active and $s_u$ is a proper sink. \hfill \emph{(proper sink)}
\end{enumerate}
\end{definition}

For $\iter^{dt}$, it is helpful to think of the successors $s_u$ as describing a fan-out $1$ dag on a line of $N$ nodes.
Active nodes are those which have some edge to a node
Then, if we require that $1$ is active, then there must be an active node which points backwards or an active node which points at an inactive node.
This can be viewed as a black-box version of the $\iterlong$ problem described earlier, as in the $\iterlong$ problem the successor values $S_u$ were described by a circuit, as opposed to being described directly by query access.

\section{Extended-Frege-Provable \texorpdfstring{$\TFNP$}{TFNP} reductions.}
\label{sec:EF_reductions}

As we have seen above, if $F(x_1, \dots, x_n) = C_1 \land \cdots \land C_m$ is an unsatisfiable CNF formula, then $\search_F$ is a total search problem, because every input assignment $x \in \bits^n$ will falsify some clause.
This means that if we have a decision-tree reduction from $\search_F$ to some search problem $R^{dt}$ which is \emph{a priori} known to be total, then the decision-tree reduction is itself a proof that $F$ is unsatisfiable.
In fact, it has been shown that highly efficient decision tree reductions from $\search_F$ to certain fixed total search problems actually \emph{completely characterize} efficient proofs in certain proof systems \cite{GoosKRS19, kamath2019, GoosH0MPRT22, DavisR23}, and moreover, any proof system satisfying a few mild conditions can be characterized in this way \cite{BussFI23}.

In this section, we consider reductions from $\search_F$ directly to classical $\TFNP$ problems, such as $\iter$ (instead of $\iter^{dt}$).
Of course, now it no longer makes sense to use black-box reductions, so instead we will follow the approach of classical mapping reductions and use boolean circuits instead.
First, we introduce an auxiliary formula which will help us define our reductions.

\begin{definition}\label{def:unsat-cnf}
    Let $F(x_1, \dots, x_n) = C_1 \land C_2 \land \cdots \land C_m$ be a CNF formula.
    The formula $\unsat_F(x_1, \dots, x_n, z_1, \dots, z_m)$ is defined to be \[\unsat_F(x, z) := \left(\bigvee_{i=1}^m z_i\right) \land \bigwedge_{i=1}^m (\neg z_i \lor \neg C_i).\]
    If $\mathcal{F} = \set{F_i}_i$ is a sequence of unsatisfiable CNF formulas, then let $\unsat_{\mathcal{F}} = \set{\unsat_{F_i}}_i$ be the corresponding sequence of formulas.
\end{definition}

When $F$ is unsatisfiable, the $\unsat_F$ formula can be viewed as a formalization of the search problem $\search_F$: given an assignment $x \in \bits^n$, output any $z$ such that $\unsat_F(x, z)$ is satisfied.
Since $F$ is unsatisfiable, such a $z$ will always exist, and hence $\forall x \exists z \unsat_F(x, z)$ is itself a tautology.
First we observe that these are easily computable as search problems.

\begin{proposition}
    If $\mathcal{F}$ is a polynomial-time uniform sequence of unsatisfiable CNF formulas, then $\unsat_{\mathcal{F}}$, interpreted as a total search problem, is in $\FP$.
\end{proposition}
\begin{proof}
    Write $\mathcal{F} = \set{F_i}_i$, and let $A$ be the polynomial-time algorithm such that $A(1^i)$ outputs the description of $F_i$.
    The polynomial-time algorithm for $\unsat_{\mathcal{F}}$ is simple: given $x$ as input, run $A(1^{|x|})$ to recover $F_{|x|}$.
    Evaluate $F_{|x|}(x)$ and, for each clause $C_i$ of $F_{|x|}$, let $z_i = C_i(x)$, and output the string $z$.
\end{proof}

A crucial distinction must be made here.
The family $\mathcal{F}$ is polynomial-time uniform, and $\unsat_{\mathcal{F}}$ is computable in polynomial time (when considered as a search problem), but both of these facts should not be interpreted as saying there are polynomial-time verifiable proofs that the families are unsatisfiable.

Next, we define a notion of non-uniform mapping reduction from $\search_F$ to an arbitrary $\TFNP$ relation $R$ using the $\unsat_F$ formulas.

\begin{definition}
\label{def:non-uniform mapping reduction}
    Let $R \in \TFNP$ and let $F(x_1, \dots, x_n) = C_1 \land \cdots \land C_m$ be an unsatisfiable CNF formula.
    Let $\ell(n)$ denote the polynomial upper bound on the length of any solution to $R$ on inputs of length $n$.
    A \emph{non-uniform mapping reduction} from $\search_F$ to $R$ is given by two circuits $C : \{0,1\}^n \rightarrow \{0,1\}^s, D : \bits^n \times \bits^{\ell(s)} \rightarrow \bits^m$, such that
    \begin{equation}\label{eq:nu-reduction}
        R(C(x), y) \implies \unsat_F(x, D(x, y)).
    \end{equation} 
\end{definition}


Observe that we have no efficient way of verifying that two circuits $(C, D)$ actually satisfy the previous definition.
For example, given $\search_F$, one could consider the following ``trivial'' polynomial-time algorithm that solves it: given $x \in \{0,1\}^n$, test the clauses of $F$ one-by-one on $x$, and output the first one that is false.
If no clause is false, then loop forever.
If $F$ is unsatisfiable, then we are guaranteed that this algorithm will halt in polynomial time.
However, \emph{proving} that this algorithm halts is equivalent to proving that $F$ is unsatisfiable in the first place.
To turn this notion into a propositional proof system, we will supply an Extended Frege proof that the mapping reduction is correct.
In order to do so, we will need to introduce a propositional formalization of Equation \cref{eq:nu-reduction}.

Recall (\cref{def:tfnp-cnf-encoding}) that if $R \in \TFNP$ then $V_{R, n, \ell}$ is a polynomial-size circuit which is equivalent to $R$ on inputs of length $n$ and $\ell$.
First, we define a formula $\mathrm{Reduction}_R$ which captures the truth of \cref{eq:nu-reduction}.

\begin{definition}
    Let $n, m$ be positive integers, let $R \in \TFNP$, and let $F(x_1, \dots, x_n) = \wedge_{i=1}^m C_i$ be a CNF formula.
    Let $\ell$ be a polynomial bound on the certificate length of $R$.
    If $R \in \TFNP$ and $F(x_1, \dots, x_n) = \wedge_{i=1}^m C_i$, then $\Reduction_{R, n, m, S, s}(C, D, F)$ is the boolean formula defined as follows.
    We introduce several helper formulas which can easily be represented in CNF.
    \begin{itemize}
        \item $\isCkt_S(C)$ is satisfied iff the variables in $C$ encode a valid boolean circuit with $S$ gates.
        \item $\Eval_{n, S, s}(C, x, t)$ is satisfied iff $C$ encodes a boolean circuit with $n$ inputs, $S$ gates, and $s$ outputs, and $t \in \bits^s$ is the output of $C$ on input $x \in \bits^n$.
    \end{itemize}
    Then $\Reduction_{R, n, m, S, s}(C, D, F):=$ \[\isCkt(C) \land \isCkt(D) \land \Eval(C, x, x') \land \Eval(D, (x, y), z) \land \Eval(V_R, x', y) \rightarrow \unsat_F(x, z).\]
    For readability, we have omitted parameter subscripts above, and we will continue to do so when those parameters are clear from context. 
    The circuit $C$ has $n$ inputs, $S$ gates, and $s$ outputs, and the circuit $D$ has $n+\ell(s)$ inputs, $S$ gates, and $m$ outputs.
\end{definition}

By definition, for any circuits $C$ and $D$ and any CNF formula $F$, we have that $(C, D)$ form a non-uniform mapping reduction from $\search_F$ to $R$ if and only if $\Reduction_R(C, D, F)$ is a tautology.
So, with this definition in hand, we can define a \emph{provable} mapping reduction from $\search_F$ to $R$.
Now, along with the circuits $C, D$, defining a mapping reduction, we supply an Extended Frege proof $\Pi$ of $\Reduction_R(C, D, \Pi)$.

\begin{definition}\label{def:ef-mapping-r} (Defining $\langle EF,R\rangle$)
    For any $R \in \TFNP$, the propositional proof system $\langle EF, R \rangle$ is defined as follows.
    For any unsatisfiable CNF formula $F(x_1, \dots, x_n) = C_1 \land \cdots \land C_m$, a refutation of $F$ in $\langle EF, R \rangle$ is given by an \emph{$EF$-provable reduction from $\search_F$ to $R$}.
	Formally, such a reduction is a tuple $(C, D, \Pi)$, where $(C, D)$ are circuits computing a non-uniform mapping reduction from $\search_F$ to $R$, and $\Pi$ is an Extended Frege proof of a boolean formula $\mathrm{Reduction}_R(C, D, F)$. 
\end{definition}

Observe that an $\langle EF, R \rangle$ proof can be easily verified in polynomial time, since we just need to verify the Extended Frege proof $\Pi$ of $\Reduction_R(C, D, F)$.
Since $R$ is a total search problem, if $(C, D)$ are a valid mapping reduction then $\Pi$ proves that $\search_F$ is total, and hence $F$ must be unsatisfiable.
(Note that Extended Frege itself may \emph{not} be able to prove that $R$ is total, and indeed later we will be interested in relations $R$ where it seems this is not the case.)
To see that $\langle EF, R \rangle$ is a complete proof system we show that it polynomially simulates Extended Frege.

\begin{proposition}
\label{prop:EF_simulates_FP}
    For any $R \in \TFNP$, $\langle EF, R \rangle \geq_p EF$.
\end{proposition}
\begin{proof}
    We give a sketch of a proof that can be easily formalized in Extended Frege. 
    Let $\Pi$ be an Extended Frege refutation of $F(x_1, \dots, x_n) = C_1 \land \cdots \land C_m$. 
    Let $C$ be a circuit which, given $x$, ignores it and hard-codes an instance of $R$ with a fixed valid solution $o$.
    Let $D$ be the circuit which, given $(x, o)$, ignores $o$, evaluates $F(x)$, and outputs the string $z \in \bits^m$ defined by $z_i = 1$ iff $C_i(x) = 0$. 
    From the definition of $D$, it is immediate that $\neg z_i \lor \neg C_i(x)$ holds for all $i \in [m]$, and so we prove that for all $x$, $D(x, o) \neq 0^m$ by contradiction.
    If $D(x, o) = 0^m$, then by the definition of $D$, we would have $C_i(x) = 1$ for all $i \in [m]$, and hence we can prove $D(x, o) = 0^m$ implies $F$ is satisfiable.
    But $\Pi$ is an Extended Frege refutation of $F$, which leads to a final contradiction.
\end{proof}

As a second calibration for its strength, it turns out that a converse holds provided Extended Frege can prove that $R$ is inside of $\FP$.

\begin{definition}
    If $R \in \TFNP$, then we say $R$ is \emph{$EF$-provably inside of $\FP$} if for every input length $n$, there is a polynomial-size circuit $A$ and a polynomial-size Extended Frege proof of the formula $\Eval(A, x, y) \land \Eval(V_R, (x, y), b) \land b$, where $V_R$ is understood to encode the polynomial-time verifier for the relation $R$.
\end{definition}

\begin{proposition}
\label{prop:FP_simulates_EF}
    If $R \in \TFNP$ is $EF$-provably in $\FP$, then $\langle EF, R \rangle \equiv_p EF$.
\end{proposition}
\begin{proof}
    We again sketch the proof and note it can be easily formalized.
    From the previous proposition, we have that $\langle EF, R \rangle \leq_p EF$.
    For the converse, consider an $\langle EF, R \rangle$ proof $(C, D, \Pi)$ of $F$.
    Let $A$ be a polynomial-size circuit solving $R$ such that $\Psi$ is an Extended Frege proof of $\Eval(A, x, y) \land \Eval(V_R, (x, y), b) \land b$, which is equivalent to $V_R(x, A(x))$.
    Then $\Pi$ proves the statement \[V_R(C(x), y) \rightarrow \unsat_F(x, D(x, y)),\] so by substituting $C$ into $\Psi$, we can prove the statement $V_R(C(x), A(C(x)))$ in Extended Frege, and hence also $\unsat_F(x, D(x, A(C(x)))).$
    From this we can derive $\neg F$ in Extended Frege directly.
\end{proof}

As we have mentioned in the introduction, later we will see examples of total search problems $R \in \FP$ but for which if $R$ was \emph{$EF$-provably} in $\FP$, then surprising consequences will hold: $EF$ can polynomially simulate $G_1$ over propositional tautologies.

\section{Implicit Resolution and \texorpdfstring{$G_1$}{G11}}\label{sec:G1}

In this section we will use provable mapping reductions in order to establish the following theorem. 

\begin{theorem}
\label{thm:Res-G1}
	$[EF, \mathrm{Resolution}]$ and $G_1$ are polynomially equivalent.
\end{theorem}

This theorem will be proven in two parts over the next subsections:

\begin{theorem}\label{claim:res_eq_pls}
	$\langle EF, \iter\rangle \equiv_p [EF, \mathrm{Resolution}]$.
\end{theorem}

\begin{theorem}\label{claim:g1_to_mapping}
	$G_1 \equiv_p \langle EF, \Iter\rangle$.
\end{theorem}

\subsection{\texorpdfstring{$\langle EF, \iter \rangle$}{<EF, Iter>} is equivalent to \texorpdfstring{$[EF, \mathrm{Resolution}]$}{[EF, Resolution]}}
\label{sec:PLSRes}

In this section we prove \cref{claim:res_eq_pls}, showing  that $\langle EF, \iter \rangle$ is polynomially-equivalent to the implicit Resolution system $[EF, \mathrm{Resolution}]$.
The intuition for this proof is inspired by the known connection between Resolution proofs and reductions to $\iterlong^{dt}$.
This connection, while having roots in results in bounded arithmetic, is perhaps best exemplified by the result of Kamath \cite{kamath2019} (cf.~\cref{thm:kamath-res-iter}), who proved that low-depth decision-tree reductions from $\search_F$ to $\iter^{dt}$ correspond directly to low-width Resolution refutations of $F$.
In some sense, our proof of \cref{claim:res_eq_pls} can be viewed as a formalization in bounded arithmetic of \cref{thm:kamath-res-iter}.

\subsubsection{Setup for Proof of \texorpdfstring{\cref{claim:res_eq_pls}}{Theorem \ref{claim:res_eq_pls}}}\label{sec:setup}

In order to prove \cref{claim:res_eq_pls} we need to show how to simulate $\langle EF, \iter \rangle$ proofs by $[EF, \mathrm{Resolution}]$ proofs and \emph{vice-versa}.
This will require formalizing proofs in Extended Frege, which can be quite technical, so instead we opt to formalize these results in the bounded arithmetic theory $\V^1$ (actually, $\V^1(\VPV)$), and then use the known propositional translations to convert $\V^1$ proofs into Extended Frege proofs.
See Appendix \ref{appendix:BA} for more details on the theories $\V^1$ and $\V^1(VPV)$. 

First, we introduce some formulas in the language of $V^{1}(\VPV)$ which capture the statements $\Reduction_\iter$ and $\ImpProof_{\Res}$.
We recall that $\V^1(\VPV)$ has variables of two sorts: \emph{number sort} variables, represented by lower-case letters $x, y, z, \dots$, and \emph{string sort} variables, represented by uppercase letters $X, Y, Z, \dots$. 
String sort variables $X$ can be indexed into by number sort variables $i$, where $X(i)$ is true if the $i$th entry of $X$ is $1$, and the first entry of the string is $0$.
We note that the string sort variables represent \emph{arbitrary} finite-length strings, and so $|X|$, which is the length of $X$, is defined to be $\max \{i : X(i) = 1\} + 1$.
Hence we can also index into $X$ beyond $|X|$, but the response is always $0$.
Finally, $\V^1(\VPV)$ has a function symbol $A(\vec x, \vec X)$ for every polynomial-time string algorithm $A$. 
We refer to Appendix \ref{appendix:BA} or the monograph \cite{CN2010} for more technical details regarding two-sorted theories of bounded arithmetic.

\paragraph{Circuit Encodings.}

We encode a circuit $C$ with $s$ gates on $n$ inputs by a pair $(n, C)$, where $n$ is the number of inputs and where the gates are numbered from $0, 1, \dots, s+n+1$ 
To simplify our encoding, we express our circuits over the $\{\land, \neg\}$ basis, writing $\lor$ as $\neg \land \neg$. 
Let $\beta(n, m)$ denote the usual pairing function\footnote{In \cite{CN2010} the notation $\langle n, m \rangle$ is used. Since we use $\langle \rangle$ for provable mapping reductions, we opt to use $\beta$ instead.} \cite{CN2010}.
The string $C$ has length $\beta(2(s+n), \beta((s+n), (s+n)))$, where the substring $C^{[0]}$ is of length $2(s+n)$ and encodes the gate sequence of the circuit, and the string $C^{[1]}$ is a two-dimensional string of length $\beta(s+n, s+n)$ and encodes the edge relation of the circuit.
The $s+n$ gate labels are encoded by two bits each, where the first $n$ gates encode the input variables, and the next $s$ gates are either $\land$ gates, $\lor$ gates, or the constants $0$ or $1$.
More formally:
\begin{itemize}
	\item For $i = 0, \dots, n-1$ we set $C^{[0]}(2i)C^{[0]}(2i+1) = 00$, representing that this is an input gate.
	\item For $i = n, \dots, s+n-1$, we break into cases:
		\begin{itemize}
			\item if $C^{[0]}(2i)C^{[0]}(2i+1) = 00$, then this gate is the constant $0$.
			\item If $C^{[0]}(2i)C^{[0]}(2i+1) = 01$, then this gate is the constant $1$.
			\item If $C^{[0]}(2i)C^{[0]}(2i+1) = 10$, then this is an $\land$ gate.
			\item If $C^{[0]}(2i)C^{[0]}(2i+1) = 11$, then this a $\neg$ gate.
		\end{itemize}
\end{itemize}
Finally, for each $i$ and $j$, $C^{[1]}(i,j) = 1$ iff there is a wire connecting the output of gate $i$ into the input of gate $j$. 
Of course, the input variables and constants $0, 1$ do not receive any inputs.
For the sake of brevity, we will often write $C$ instead of the tuple $(n, C)$.
Let $\mathrm{Size}(n, C) = |C^{[0]}|/2$ denote the number of gates in the circuit $C$, including the number of input variables.
Let $\isCkt(n,C)$ denote the relation that is true iff the input properly encodes a circuit with $s = \mathrm{Size}(n, C)$ gates in total and $n$ input gates, and note that this relation is definable in $V^0$.
If $C$ is a circuit of size $s$ encoded as $(n, C)$, then for $o < \mathrm{Size}(C)$ we let $\Eval(n, o, C, X)$ denote the polynomial-time function which evaluates $C$ on the string $X$ and outputs the values of the final $o$ gates as a single string.

%

\paragraph{Encoding Resolution Proofs.}

Let $F(x_1, \dots, x_n) = C_1 \land \cdots \land C_m$ be an unsatisfiable CNF formula, encoded by strings $P, N \in \bits^{m \times n}$ as described above.
The validity of an implicit Resolution refutation of $F$ is encoded by a formula \[\mathrm{ImpProof}_{\Res}(n, m, s, C, P, N)\] which we define now.
The string $C$ will encode a circuit that implicitly encodes a Resolution refutation of $F$, following the definition given in \cref{sec:implicit-proof-systems}.
Each line of the proof is described by a string of length $\ell^* := 2n + 2 + 2s$, which together encode a clause, how the line was derived, and pointers to earlier clauses.
As a function, $C: \bits^s \rightarrow \bits^{\ell^*}$, taking a string as input which indexes a line of the proof, and outputs a string encoding the entire line.

We represent $\mathrm{Proof}_{\Res}$ as a conjunction of polynomial-time computable predicates, each of which is a symbol in the language of $\V^1(\VPV)$.
Below, the formula $F$ is represented by the pair of strings $(P, N)$.
We say a line is \emph{active} if the line is not disabled in the proof.
Formally, we consider the polynomial-time predicate symbols:
\begin{itemize}
    \item $\mathrm{Active}(n, s, C, Y)$ is true if the clause at line $C(Y)$ is active.
    \item $\mathrm{Empty}(n, s, C, Y)$ is true if the clause at line $C(Y)$ is empty.
    \item $\mathrm{isWeaken}(n, s, C, Y)$ is true if the line $C(Y)$ of the proof is a valid weakening of an earlier, active line.
    \item $\mathrm{isResolve}(n, s, C, Y)$ is true if the line $C(Y)$ is a valid Resolution of earlier active lines.
    \item $\mathrm{isAxiomW}(n, m, s, i, P, N, C, Y)$ is true if the clause at line $C(Y)$ is a weakening of the $i$th clause of $F$.
\end{itemize}
If $m$ is a number term then we let $\mathrm{Dyad}(m)$ denote the string encoding $m$ in dyadic notation, and note that $\mathrm{Dyad}(m)$ is a polynomial-time computable function and its basic properties can be formalized in $\V^1$ \cite{CN2010}.

\begin{definition}\label{def:imp_proof}
The formula $\mathrm{ImpProof}_{\Res}(n, m, s, C, P, N)$ is the conjunction of the following formulas:
\begin{itemize}
    \item $\isCkt(n, C)$
    \item $\mathrm{Active}(n, s, C, 1^s) \land \mathrm{Empty}(n, s, C, 1^s)$
    \item $\forall Y < s: \mathrm{Active}(n, s, C, Y) \rightarrow \mathrm{isWeaken}(n, s, C, Y) \lor \mathrm{isResolve}(n, s, C, Y) \lor \exists i < m: \mathrm{isAxiomW}(n, m, s, i, P, N, C, Y)$.
\end{itemize}
\end{definition}

\paragraph{Encoding Mapping Reductions.}

We now define a formula \[\mathrm{Reduction}_\iter(n, m, s, C, D, P, N)\] in the language of $\V^1(\VPV)$ which verifies if the circuits $C$ and $D$ are a non-uniform mapping reduction from $\search_F$ to $\iter$.
For this, we introduce a new formula \[\unsat(n, m, P, N, X, Z),\] which can be viewed as the $\Delta^B_0$-formula encoding  $\unsat_F$  from \cref{def:unsat-cnf}.
This is indeed closely related to the propositional formula $\unsat(C, D, F)$ we introduced in \cref{def:unsat-cnf}, albeit with a slight change in the encoding for the formula $F$.
In this formula, $(P, N)$ are strings of length $\bits^{m \times n}$ encoding a CNF formula $F$ with $n$ variables and $m$ clauses, $X$ is a string encoding an input to $F$, and $Z$ is a string encoding indices to false clauses of $F$ on input $X$.
In particular, applying the propositional translation $||\unsat(n, m, P, N, X, Z)||_{n, m}$ will give essentially the same formula as $\unsat(C, D, F)$, albeit with the formula $F$ encoded by lists of free variables $P$ and $N$.
Formally,
\begin{multline*}
    \unsat(n, m, P, N, X, Z) := (\exists i < m: Z(i)) \land \\
    \forall i < m: Z(i) \rightarrow (\forall j < n: (P^{[i]}(j) \rightarrow \neg X(j)) \land (N^{[i]}(j) \rightarrow X(j)))
\end{multline*}
Using this formula we can define the reduction formula.
\begin{definition}\label{def:reduction-v1vpv}
    Let $\iter$ denote the polynomial-time function symbol encoding the polynomial-time verifier for the $\iterlong$ problem.
    The formula $\Reduction_\iter(n, m, s, C, D, P, N)$ is defined to be the conjunction of the formulas:
    \begin{itemize}
        \item $\isCkt(n, C)$,
        \item $\isCkt(n+s, D)$,
        \item and the formula
        \begin{multline*}
            \forall X < n, Y < s: \iter(s, \Eval(n, s, C, X), Y) \rightarrow\\
            \unsat(n, m, P, N, X, \Eval(n+s, m, D, X \circ Y))
        \end{multline*} 
        where $\circ$ denotes string concatenation.
    \end{itemize}
\end{definition}
We note for the reader that $C$ is a circuit which itself \emph{outputs} the description of a circuit, since the input to $\iter$ is itself a circuit.

\paragraph{Statement of Main Lemmas and Proof of \cref{claim:res_eq_pls}.}

Finally, using the notation introduced in the previous sections, we can now state the main lemmas which formalize the transformations between implicit Resolution proofs and non-uniform reductions inside of $\V^1(\VPV)$.

\begin{lemma}
\label{thm:V1ProvesPLSRes}
    There are polynomial-time algorithms $A$ and $B$ such that $\V^1(\VPV)$ proves
    \[\mathrm{ImpProof}_{\Res}(n, m, s, C, P, N) \rightarrow \mathrm{Reduction}_\iter(n, m, s, A(\vec x, \vec X), B(\vec x, \vec X), P, N), \]
    where $\vec x = (n, m, s)$, and $\vec X = (C, P, N)$.
\end{lemma}

\begin{lemma}
\label{thm:formalizationPLSResConv}
    There is a polynomial-time algorithm $R$ and a number term $t$ such that $\V^1(\VPV)$ proves
    \[ \mathrm{Reduction}_\iter(n, m, s, C, D, P, N) \rightarrow \ImpProof_{\Res}(n, m, t(\vec y, \vec Y), R(\vec y, \vec Y), P, N), \] where $\vec y = (n, m, s)$ and $\vec Y = (C, D, P, N)$.
\end{lemma}

The proofs of these two lemmas will take up the rest of the section.
However, assuming \cref{thm:V1ProvesPLSRes} and \cref{thm:formalizationPLSResConv}, we can now prove \cref{claim:res_eq_pls}. 

\begin{proof}[Proof of \cref{claim:res_eq_pls}]
    This proof is an exercise in the use of propositional translations, so we sketch it and leave a fully formal proof to the reader.

    First, we show that $[EF, \mathrm{Resolution}]$ $p$-simulates $\langle EF, \iter \rangle$.
    Suppose that $(C, D, \Pi)$ is an $\langle EF, \iter \rangle$ refutation of a CNF formula $F(x_1, \dots, x_n) = C_1 \land \cdots \land C_m$, and so $\Pi$ is an Extended Frege proof of the propositional formula $\Reduction_{\iter}(C, D, F)$.
    Let $(P^F, N^F) \in \bits^{m \times n}$ be strings encoding the CNF $F$. 
    By propositionally translating the statement in \cref{thm:formalizationPLSResConv}, we have for any parameters $n', m', s' \in \mathbb{N}$, there are polynomial-size Extended Frege proofs of the implication
    \[ ||\mathrm{Reduction}_\iter(\vec y, \vec Y)||_{n', m', s'} \rightarrow ||\ImpProof_{\Res}(n, m, t(\vec y, \vec Y), R(\vec y, \vec Y), P, N)||_{n', m', s'}, \] 
    where we have used the notation $\vec y$ and $\vec Y$ from \cref{thm:formalizationPLSResConv}.
    From $\Reduction_{\iter}(C, D, F)$, Extended Frege can prove the formula \[||\Reduction_{\iter}(n, m, s, C, D, P^F, N^F)||_{n, m, s}\] obtained by propositionally translating the formula from \cref{def:reduction-v1vpv}.
    Hence Extended Frege can prove \[||\ImpProof_{\Res}(n, m, t(\vec y, \vec Y), R(\vec y, \vec Y), P^F, N^F)||_{n', m', s'}\] and so we have proved that the circuit encoded by $R$ encodes an implicit Resolution refutation of $F$.

    The converse direction is similar. 
    Suppose that $(\Pi, C)$ is an $[EF, \mathrm{Resolution}]$ refutation of $F$, so $\Pi$ is an Extended Frege proof of the propositional translation \[||\ImpProof_{\Res}(n, m, s, C, P^F, N^F)||_{n, m, s}\] for appropriately chosen $n, m, s \in \NN$.
    By applying a similar argument as before, we propositionally translate \cref{thm:V1ProvesPLSRes} to attain an Extended Frege proof of the implication in that lemma, and hence we can derive in Extended Frege the formula \[||\mathrm{Reduction}_{\iter}(n, m, s, A(\vec x, \vec y), B(\vec x, \vec Y), P^F, N^F)||_{n, m, s} \]
    in polynomial size.
    From this formula, and using the fact that $P^F$ and $N^F$ encode $F$, we can derive the propositional version of $\Reduction(A, B, F)$ in size polynomial in $|C|$ and $|\Pi|$ which used in the definition of $\langle EF, \iter \rangle$.

\end{proof}


\subsubsection{Proofs of Main Lemmas}

In the remainder of this section we prove \cref{thm:V1ProvesPLSRes} and \cref{thm:formalizationPLSResConv}, and we begin by proving \cref{thm:V1ProvesPLSRes}.
Our proof is inspired by the equivalence between black-box total search problems efficiently reducible to $\iter^{dt}$ and low-width Resolution refutations, as observed by \cite{kamath2019}, but analogous results were observed in the bounded arithmetic literature, see e.g.~\cite{krajivcek2001weak}.
Indeed, our proof can, in some sense, be viewed as a \emph{formalization} in $\V^1(\VPV)$ of this theorem, although in fact moving to the white-box setting will simplify some steps and complicate others.

\begin{proof}[Proof of \cref{thm:V1ProvesPLSRes}]
    We argue in $\V^1(\VPV)$.
    Assume $\ImpProof_{\Res}(n, m, s, C, P, N)$, and let $\vec x = (n, m, s)$ and $\vec X = (C, P, N)$.
    To remember the types of the various objects, we recall that $(P, N)$ are strings encoding a CNF formula $F(x_1, x_2, \dots, x_n) = C_1 \land \cdots \land C_m$, and $C$ is a circuit implicitly encoding a Resolution refutation of $F$.
    As input, $C$ takes a string $Y \in \bits^s$ and outputs a corresponding line of the proof.
    The algorithms $A$ and $B$ that we define will \emph{output} circuits $S$ and $G$ encoding an instance of $\iter$, such that any solution in the $\iter$ instance can be used to recover a false clause of $F$.

    We now describe the circuits $S$ and $G$ which will be output by the algorithms $A$ and $B$, respectively.
    The circuit $S$, on input $X, Y$ with $X < n$ and $Y < s$, executes \cref{algo:successor}.
    \begin{figure}[htbp]
    \begin{algorithm}[H] 
        \SetAlgoLined
        \KwIn{Strings $X, Y$ with $X < n$ and $Y < s$}
        \KwOut{A string $Y' < s$}
        
        Evaluate $C(Y)$, receiving a line $\mathcal{L}$ of the proof\;
        Let $\mathcal{L} = (P^\mathcal{L}, N^{\mathcal{L}}, T^{\mathcal{L}}, L^{\mathcal{L}}, R^{\mathcal{L}})$ be the components of the line\;
        \If{$\mathrm{SATClause}(n, P^{\mathcal{L}}, N^{\mathcal{L}}, X)$ or $T^{\mathcal{L}} = 10$ or $T^{\mathcal{L}} = 11$}{
            \Return{$Y$} 
        }
        \uIf{$T^{\mathcal{L}} = 00$}{
            \tcc{$C(Y)$ is a weakening, so output the line that was weakened.}
            \Return{$L^{\mathcal{L}}$} 
        } 
        \uElseIf{$T^{\mathcal{L}} = 01$}{
            \tcc{$C(Y)$ is a resolution, so find the false clause} 
            Evaluate $C(L^Y) = \mathcal{L}'$\;
            Let $\mathcal{L}' = (P^{\mathcal{L}'}, N^{\mathcal{L}'}, T^{\mathcal{L}'}, L^{\mathcal{L}'}, R^{\mathcal{L}'})$ be the components of the line\;
            \uIf{$\neg \mathrm{SATClause}(n, P^{\mathcal{L}'}, N^{\mathcal{L}'})$}{
                \Return{$L^{\mathcal{L}}$}
            }
            \uElse{
                \Return{$R^{\mathcal{L}}$}
            }
        }
        
        \caption{Successor Function $S(n, s, X, Y)$}
        \label{algo:successor}
    \end{algorithm}
    \end{figure}

To summarize the algorithm in \cref{algo:successor}, we evaluate the circuit $C$ at the index $Y$ to get a line of the proof $C(Y)$.
If $C(Y)$ is satisfied, or disabled, or an axiom of $F$, then we output $Y$ as the successor.
Otherwise, $C(Y)$ must be falsified, active, and therefore derived by either weakening or Resolution.
We hence evaluate the lines used to derive $C$, and let the successor of $Y$ be the falsified line that was used to derive $C(Y)$.
This algorithm is computable in polynomial time, and the circuit computing $S$ can be computed from $n, s,$ and $C$ in polynomial time as well.
Given $\vec x, \vec X$, the algorithm $A$ outputs the encoding of the circuit computing the function $X \mapsto S(X, Y)$.

The circuit $G$ is even simpler.
Given $X < n, Y < s$, we evaluate $S$ at most twice to see if $Y$ is a solution of the $\iter$ instance encoded by $S$.
If $Y$ is a solution, then by construction, $S(X, Y) = Y'$ is the index of a falsified clause of $F$ on input $X$, and hence $G$ will output a string $Z < m$ indexing the clause falsified by $X$ on this input.
Similarly, the circuit $G$ can be computed from $\vec x$ and $\vec X$ in polynomial time, and so the algorithm $B$ constructs $G$ and outputs the encoding of the circuit computing the function $(X, Y) \mapsto G(X, Y)$.

Finally, we argue that $\V^1(\VPV)$ can prove \[\Reduction_{\Iter}(n, m, s, A(\vec x, \vec X), B(\vec x, \vec X), P, N).\]
By definition, $A$ outputs the code of circuit $S$ and $B$ outputs the code of circuit $G$, and hence $\V^1(\VPV)$ can prove $\isCkt(n, A(\vec x, \vec X))$ and $\isCkt(n+s, B(\vec x, \vec X)).$
Now, let $X < n$ and $Y < s$, and suppose that $\iter(s, \Eval(n, s, A(\vec x, \vec X), Y)$ holds. 
Examining the code of $A(\vec x, \vec X) = S$ and $B(\vec x, \vec X) = G$, we have that $G$ will output a string $Z$ falsifying a clause of $F$, as can be proved using a case argument on the definition of $G$ and $S$.
Hence $\unsat(n, m, P, N, X, \Eval(n+s, m, B(\vec x, \vec X), X \circ Y))$ will hold, and this completes the proof.
\end{proof}

We now prove \cref{thm:formalizationPLSResConv}.
\begin{proof}[Proof of \cref{thm:formalizationPLSResConv}]
    We again reason in $\V^1(\VPV)$.
    Assume that \[\Reduction_\iter(n, m, s, C, D, P, N)\] holds, and so $C$ and $D$ are circuits encoding a reduction from the search problem for the CNF $F$ encoded by $(P, N)$ to $\iter.$
    This means given an input assignment $X < n$, $C(X)$ will output the description of a circuit $S$ that is an instance of $\iterlong$, and given $X < n$, $Y < s$, $D(X, Y)$ will output a solution $Z$ of $\search_F$.
    For simplicity, we write $S(X, Y)$ to be the algorithm which first evaluates $C(X)$, obtaining a circuit $S$, and then evaluates $S(Y)$.

    \begin{figure}[htbp]
    \begin{algorithm}[H] 
        \SetAlgoLined
        \KwIn{Strings $X \leq n+1$, $Y < s, Y' < s$.}
        \KwOut{Strings encoding a line $\mathcal{L} = (P^{\mathcal{L}}, N^{\mathcal{L}}, T^{\mathcal{L}}, L^{\mathcal{L}}, R^{\mathcal{L}})$}

        \If{$X = 1^{n+1}, Y = 1^s, Y' = 1^s$}{
            \tcp{This is a special final ``empty'' clause to satisfy the definition of $\ImpProof_{\mathrm{Res}}$.}
            \Return{$\mathcal{L} = (\varepsilon, \varepsilon, 00, (\varepsilon, \varepsilon, \varepsilon))$} 
        }
        \uIf{$|X| = n+1$}{
            \tcp{We are in the ``first block'', and so we use the successor function $S$ to route clauses.}
            Let $X' = X[0:n-1]$ be the first $n$ bits of $X$\;
            \If{$S(X',Y) \neq Y'$}{
                \tcp{Disable the line.}
                \Return{any line with $T^{\mathcal{L}} = 11$}
            }
            \tcp{Now $S(X',Y) = Y'$}
            \uIf{$|Y'| = |Y| = 0$ or $Y' <_\mathrm{lex} Y$ or $S(X', Y') = Y'$}{
                \tcp{$Y$ is a solution, so $\unsat(n, m, P, N, X, D(X', Y))$ is satisfied and $D(X', Y)$ encodes a clause of $F$ falsified under $X'$.}
                Evaluate $D(X', Y) = Z$\;
                Let $i := \min \{j : Z(j) = 1\}$\;
                \Return{$\mathcal{L} = (P^{\overline{X'}}, N^{\overline X'}, 10, \mathrm{Dyad}(i), 0^s)$}
            }
            \uElseIf{$Y' >_{\mathrm{lex}} Y$}{
                \tcp{$Y$ is not a solution, so we just make it a weakening.}  
                Evaluate $S(X', Y') = Y''$\;
                \Return{$\mathcal{L} = (P^{\overline{X'}}, N^{\overline{X'}}, 00, (X, Y''), 0^s)$}
            }
        }
        \uElse{\tcp{$|X| < n+1$, now we are in the second block, and must embed a complete tree-like Resolution refutation.}
            \If{$|Y| > 0$ or $|Y'| > 0$}{
                \tcp{Disable the line.} 
                \Return{any line $\mathcal{L}$ with $T^{\mathcal{L}} = 11$}
            } 
            \uIf{$|X| = n$}{
                \tcp{In this case we are a weakening of an earlier line}
                Evaluate $S(X, Y') = Y''$\;
                \Return{$\mathcal{L} = (P^{\overline{X}}, N^{\overline{X}}, 00, (X \circ 1, Y''), 0^s)$}
            }
            \uElse{
                \tcp{Now $|X| < n$, so we are performing the complete tree-like Resolution refutation of size $2^n$.} 
                \Return{$\mathcal{L} = (P^{\overline X}, N^{\overline X}, 01, (X \circ 0, 0^s, 0^s), (X \circ 1, 0^s, 0^s)))$} 
            }
        }
        \caption{Implicit Refutation $H(X, Y, Y')$}
        \label{algo:implicit-proof}
    \end{algorithm}
    \end{figure}
    
    Let $\vec y = (n, m, s)$ be the number parameters and $\vec Y = (C, D, P, N)$ be the string parameters.
    We compute an implicit Resolution refutation using the algorithm $H$, whose code appears in \cref{algo:implicit-proof}, although we will give a high-level overview of the structure of the proof now, and a depiction is given in \autoref{fig:resultingRes}.
    
    Each line of the implicit Resolution proof will be indexed by strings $X, Y, Y'$ with $|X| \leq n+1$ and $|Y|, |Y'| < s$. 
    The structure of the refutation will be in two separate blocks.
    In the first block, $|X| = n+1$, and in the second block $|X| \leq n$.
    When $|X| = n+1$, let $X'$ be the first $n$ bits of $X$, and note this can be obtained by simply setting the $n$th bit of $X$ to $0$.
    A line will be active in this case if $S(X', Y) = Y'$ and either $Y' >_\mathrm{lex} Y$ or $Y' = 0$.
    If a line is active, then $X'$ will define the clause encoded at the line by the strings $(P^{\overline X'}, N^{\overline X'})$, which is the unique clause of width $n$ that is falsified by $X'$.
    (We note that for any such $X'$, we can define $P^{\overline X'}, N^{\overline X'}$ using $\Sigma^B_0$-comprehension.)
    The clauses used to deduce this will be determined entirely by the successor function $S$ from the $\iterlong$ instance.
    If $Y$ is a solution in $S$ on the input $X'$, meaning either $S(X', Y) = Y'$ and $S(X', Y') = Y'$, then we know that $X'$ falsifies some clause of $F$, and the particular clause can be recovered using the circuit $D(X', Y)$.
    In this case, we label the line as a weakening of a clause of $F$ and we are done.
    On the other hand, if $Y$ is not a solution, then $S(X', Y) = Y'$ and $S(X', Y') \neq Y'$.
    Now the line indexed by $(X', Y', S(X', Y'))$ must also be active and labelled by the clause $(P^{\overline X'}, N^{\overline X'})$, and so we declare that the line $(X', Y, Y')$ to be a (trivial) weakening of $(X', Y', S(X', Y'))$.

    \begin{figure}[ht]
\centering
    \begin{tikzpicture}[scale=1.1]
            \tikzset{inner sep=0,outer sep=3}

            \tikzstyle{a}=[inner sep=4pt,outer sep=0pt,
            draw=black!40!white, fill=Cerulean!10!white, very thick, rounded corners=6pt, align=center]
            \tikzstyle{b}=[inner sep=2pt, inner ysep=4pt,outer sep=0.5pt,
            draw=black!40!white, fill=Cerulean!10!white, very thick, rounded corners=6pt, align=center]
            \tikzstyle{c}=[inner sep=3pt, inner ysep=4pt,outer sep=0.5pt,
            draw=black!40!white, fill=Cerulean!10!white, very thick, rounded corners=6pt, align=center]
            \tikzstyle{e}=[inner sep=1pt, inner ysep=1pt,outer sep=0.5pt,
            draw=Orange!70!white, fill=yellow!10!white, dashed, thick, rounded corners=6pt, align=center]

            \tikzstyle{f}=[inner sep=5pt, inner ysep=1pt,outer sep=0.5pt,
            draw=Orange!70!white, fill=yellow!10!white, thick, rounded corners=6pt, align=center]

                \draw[f, rounded corners=2] (-2.8,0.9) -- (-1.8,0.9) -- (-1.8,-4) -- (-2.8,-4) -- cycle;

                \draw[f, rounded corners=2] (-0.5,0.9) -- (0.5,0.9) -- (0.5,-4) -- (-0.5,-4) -- cycle;

                \draw[f, rounded corners=2] (2,0.9) -- (3,0.9) -- (3,-4) -- (2,-4) -- cycle;
                
                \draw[a] (-3.7,0.4) -- (3.7,0.4) -- (0,4) --  cycle;

                \draw[a]  (-0.8,2.6) -- (-2,1.6);
                \draw[a]  (-0.8,2.6) -- (-0.7,1.6);

                \draw[a] (0.8,2.6) -- (0.6,1.6);
                \draw[a] (0.8,2.6) -- (1.95,1.6);
                
                \draw[a] (0,3.41) -- (-0.8,2.6) ;
                \draw[a] (0,3.41) -- (0.8,2.6);
                
                \node[a] (L11) at (0,3.41) {$\bot$};

                \node[a] (L21) at (-0.8,2.6) {$x_1$};
                \node[a] (L22) at (0.8,2.6) {$\overline x_1$};

                \node[a] (L32) at (-0.7,1.6) {\small $x_1 \lor x_2$};
                \node[a] (L31) at (-2,1.6) {\small $x_1 \lor \overline x_2$};

                \node[a] (L33) at (0.6,1.6) {\small $\overline x_1 \lor x_2$};
                \node[a] (L34) at (1.95,1.6) {$\overline x_1 \lor \overline x_2$};



                \draw[draw=black!40!white, very thick] plot [smooth, tension=1] coordinates  { (0,-0.5) (0.7, -2) (0,-3.5) };
                
                \node at (-0.5,1) {\large $\vdots$};
                \node at (-1.8,1) {\large $\vdots$};
                \node at (0.5,1) {\large $\vdots$};
                \node at (1.8,1) {\large $\vdots$};

                \draw[draw=black!40!white, very thick] plot [smooth, tension=1.2] coordinates  { (-2.3,0.4) (-3, -0.5) (-2.3,-1.5) };
                
                \draw[draw=black!40!white, very thick] plot [smooth, tension=1.2] coordinates  { (-2.3,-0.5) (-1.5, -1.5) (-2.3,-2.5) };

                \node[f, circle] (a1) at (-2.3,0.4) {$\overline \beta$};

                \node[a, circle] (b1) at (-2.3,-0.5) {$\overline \beta$};

                \node[a, circle] (c1) at (-2.3,-1.5) {$\overline \beta$};

                \node[a, circle] (d1) at (-2.3,-2.5) {$\overline \beta$};
                
                \node[a, circle] (e1) at (-2.3,-3.5) {$\overline \beta$};

                \draw[a] (c1)--(d1);
                \draw[a] (d1)--(e1);

                \node[f, circle] (a) at (0,0.4) {$\overline \alpha$};

                \node[a, circle] (b) at (0,-0.5) {$\overline \alpha$};

                \node[a, circle] (c) at (0,-1.5) {$\overline \alpha$};

                \node[a, circle] (d) at (0,-2.5) {$\overline \alpha$};
                
                \node[a, circle] (e) at (0,-3.5) {$\overline \alpha$};

                \node[f, circle] (a2) at (2.5,0.4) {$\overline \gamma$};

                \node[a, circle] (b2) at (2.5,-0.5) {$\overline \gamma$};

                \node[a, circle] (c2) at (2.5,-1.5) {$\overline \gamma$};

                \node[a, circle] (d2) at (2.5,-2.5) {$\overline \gamma$};

                \draw[a] (d2) -- (0.5,-4.5);
                
                \node[a, circle] (e2) at (2.5,-3.5) {$\overline \gamma$};

                \draw[a] (a) -- (b);

                \draw[a] (c) -- (d);

                \node[a] (cli) at (0.5,-4.5) {$C_i$};

                \node[] () at (1.7,-4.5) {$\ldots$};

                \node[] () at (-0.5,-4.5) {$\ldots$};

                \node[a] (cl1) at (-2.7,-4.5) {$C_1$};
                \node[a] (cl2) at (-1.5,-4.5) {$C_2$};
                
                \node[a] (clm) at (3,-4.5) {$C_m$};

                \draw[a] (d) -- (cl2);
                
                \draw[a] (e) -- (cli);
                \draw[a] (e1)--(cl2);

                \draw[a] (a2) -- (b2);
                \draw[a] (b2) -- (c2);
                \draw[a] (c2) -- (d2);
                \draw[a] (d2) -- (cli);
                \draw[a] (e2) -- (clm);

                \node[] at (-5,-4.5) {$F$};
                \node[] at (-5,-2) {Block 1};
                \node[] at (-5,2) {Block 2};

            \end{tikzpicture} 
            \caption{The structure of the implicit Resolution proof of $F=C_1 \wedge \ldots \wedge C_m$. The second block is a complete tree-like resolution refutation. Each leaf of this complete proof is labeled with a clause $\overline \alpha$ which is the negation of (the conjunct representing) a total assignment $\alpha \in \{0,1\}^n$; this constitutes the first block. At that leaf, the reduction from $\search_F$ to $\Iter$ (marked in yellow) is run on the total assignment $\alpha$ and each node is labeled with $\overline \alpha$. The correctness of the $\Iter$ instance guarantees that we arrive at a clause falsified by $\alpha$.}
            \label{fig:resultingRes}
\end{figure}
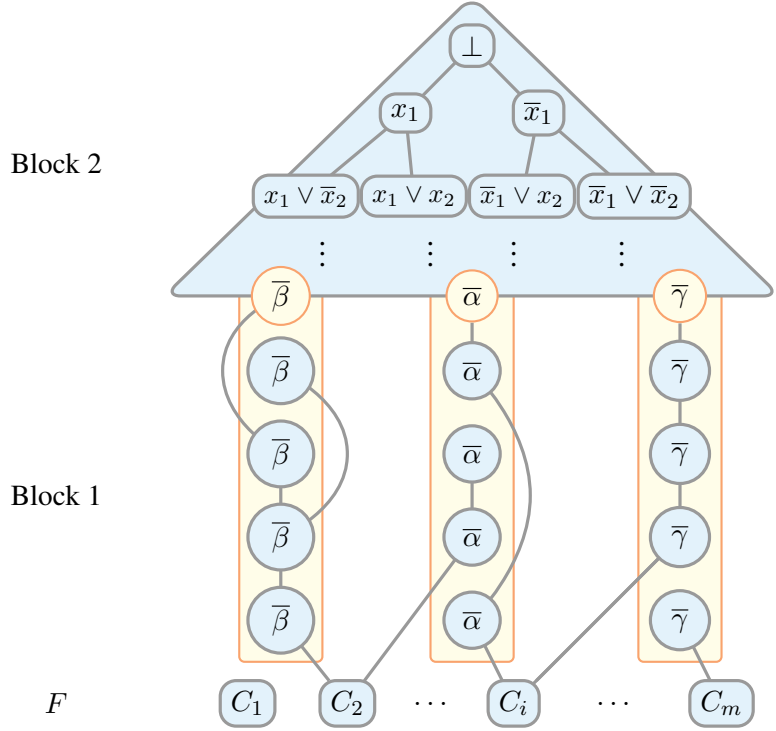

    In the second block, $|X| \leq n$.
    Now a line is active only if $|Y| = |Y'| = 0$.
    Here we simply plant the complete tree-like Resolution refutation of size $2^n$.
    The leaves of the refutation will be labelled by strings $|X| = n$, and we similarly declare these to be weakenings of corresponding clauses in the first block, and pointers to the internal clauses can be calculated quite easily.

    Now, we let $R(\vec y, \vec Y)$ denote the algorithm which outputs the circuit encoding $H$ on these input lengths, and we let $t(\vec y, \vec Y) = n+1 + 2s$.
    The algorithm $H$ is polynomial-time computable and so $R$ is a polynomial-time computable algorithm.
    Let us now argue that we can prove $\ImpProof_{\mathrm{Res}}(n, m, t(\vec y, \vec Y), R(\vec y, \vec Y), P, N)$ in $\V^1(\VPV)$.

    From the definition of $\ImpProof_{\mathrm{Res}}$, and from the fact that $t(\vec y, \vec Y) = t = n+2s+1$, we need to prove the following formulas:
    \begin{itemize}
        \item $\isCkt(t, R(\vec y, \vec Y))$
        \item $\mathrm{Active}(n, t, R(\vec y, \vec Y), 1^{t}) \land \mathrm{Empty}(n, t, R(\vec y, \vec Y), 1^{t})$
        \item The formula
        \begin{multline*} \forall Y < t: \mathrm{Active}(n, t, R(\vec y, \vec Y), Y) \rightarrow \\
        \mathrm{isWeaken}(n, t, R(\vec y, \vec Y), Y) \lor \mathrm{isResolve}(n, t, R(\vec y, \vec Y), Y) \lor \\
        \exists i < m: \mathrm{isAxiomW}(n, m, t, i, P, N, R(\vec y, \vec Y), Y).
        \end{multline*} 
    \end{itemize}
    The fact that $R$ outputs a circuit is easily provable in $\V^1(\VPV)$.
    For the remaining statements we must analyze the code of $H$.

    To prove \[\mathrm{Active}(n, t, R(\vec y, \vec Y), 1^{t}) \land \mathrm{Empty}(n, t, R(\vec y, \vec Y), 1^{t}),\] we observe this follows from the first ``If'' statement of the description of \cref{algo:implicit-proof} -- the line is active and empty by definition.

    Similarly we can prove the formula
    \begin{multline*} \forall Y < t: \mathrm{Active}(n, t, R(\vec y, \vec Y), Y) \rightarrow \\
    \mathrm{isWeaken}(n, t, R(\vec y, \vec Y), Y) \lor \mathrm{isResolve}(n, t, R(\vec y, \vec Y), Y) \lor \\
    \exists i < m: \mathrm{isAxiomW}(n, m, t, i, P, N, R(\vec y, \vec Y), Y).
    \end{multline*} 
    using the code of \cref{algo:implicit-proof} and from the assumption that $\Reduction_{\iter}$ holds.
    If a line is active, then by inspection of the algorithm, the step is either a weakening (with tag $00$), a Resolution (with tag $01$), or a weakening of an axiom (with tag $10$).
    In the first two cases, the definition of \cref{algo:implicit-proof} shows that the output lines are correct formed, and hence $\mathrm{isWeaken}$ and $\mathrm{isResolve}$ will hold for these lines.
    In the third case, the line indexed is a solution, and this can only happen if Lines 13-17 execute in \cref{algo:implicit-proof}.
    In this case, we must be sure that the output of $D(X', Y) = Z$ actually indeces into a false clause of the formula under the assignment $X'$.
    Since we have assumed that $\Reduction_{\iter}$ holds, we have (cf.~\cref{def:reduction-v1vpv}) that if $C$ encodes a solution on the input $X', Y$, then the string output by $D(X', Y)$ is indeed an false clause of $F$.
    Hence we can conclude that this is a correct axiom weakening of a clause of $F$. 
\end{proof}

\subsection{\texorpdfstring{$G_1$}{G1} is p-equivalent to  \texorpdfstring{$\langle EF, \Iter\rangle$}{<EF, \Iter>}}\label{sec:g1_imp_res}

In this section we prove \Cref{claim:g1_to_mapping}. Our proof will rely on the well known $\TV^1$/$G_1$-witnessing theorem of \cite{BussK94, skelley_g1, BeckmannB09}, stating that the $\forall\Sigma^B_1$-consequences of $\TV^1$ are witnessed by $\PLS$ functions. This was originally proved in the uniform setting by Buss and Kraj\'i\v cek \cite{BussK94} for the first order $\mathsf{T}^{1}_2$, and later reproved by Cook and Nguyen for $\TV^1$ \cite{CN2010}. A follow-up work of Beckmann and Buss \cite{BeckmannB09} generalized these witnessing theorems for $\TV^i$, $i\geq 1$.

A priori, there is little similarity between $G_1$ and $\langle EF, \Iter \rangle$. However some intuition for the proof can be built by juxtaposing this result with the $\TFNP^{dt}$ characterizations of Tree-like Resolution with $\FP^{dt}$, and bounded-width Resolution with $\Iter^{dt}$. The reason that bounded-width resolution is not similarly characterized by efficient decision-tree protocols is due to the DAG-like nature of the proofs. This same issue precisely identifies the jump from $EF$ to $G_1$. 



Showing $\langle EF, \Iter\rangle \leq_p G_1$ will require reducing the $G_1$ witnessing problem to $\Iter$, provably in Extended Frege. To do this, we formalize a $G_1$-witnessing theorem in Extended Frege by propositionally translating an efficient $\TV^1$-witnessing theorem of \cite{CN2010}. To show $G_1 \leq_p \langle EF, \Iter\rangle$, we prove that $\TV^1$ proves the reflection principle for $\langle EF, \Iter\rangle$, which will amount to proving the reflection principle of $EF$, as well as the totality of $\Iter$.

We begin with some definitions.
\begin{definition}
    The class $\Sigma^q_0 = \Pi^q_0$ consists of all (quantifer-free) propositional formulas. 
    The class of $\Sigma^q_1$ formulas are of the form $\exists x_1,\ldots,x_n  A(x_1,\ldots,x_n,y)$ where $A(x,y) \in \Sigma^q_0$, and $y=y_1,\ldots,y_m$.
    \end{definition}

The system $G$ is a proof system for quantified Boolean formulas, that natural extends the sequent calculus proof system for $\Sigma^q_0$ formulas. 
See Appendix $A$ for a formal description of $G$.
The subsystem $G_1$ is $G$ but with the restriction that all formulas occurring in all sequents of the proof are $\Sigma^q_1$ formulas.


\begin{definition}
    [Witnessing for $G_1$] Let $F= \exists y \,A(x,y)$, where $A$ is a $\Sigma^q_0$ formula, and let $\Pi$ be a $G_1$ proof of $F$. On input $F$, $\Pi$, and assignment $\alpha \in \{0,1\}^n$ to the $x$-variables,  the \emph{$G_1$ witnessing search problem} asks to find  an assignment to the $y$-variables such that $A(x,y)$ holds. If $\Pi$ is not a valid $G_1$-proof of $F$, then output $\bot$. 
\end{definition}
For simplicity, we will often write $\mathsf{Wit}(F,\Pi)$.
We use the following theorem of Cook and Nguyen, which proves $\TV^1$-witnessing in $\V^1$.
\begin{theorem}[Uniform $\TV^1$ Witnessing, \cite{CN2010}]\label{thm:eff_tv1_witnessing}
    Suppose that $\TV^1 \vdash \psi$ for $ \psi \equiv \forall x, X\,\exists y,Y\, \varphi(x,X,y,Y)$, with $\varphi \in \Sigma^B_0$. Then there is an $\Iter$ instance defined by its $\VPV$-predicate \emph{graph}, $F(x,X,y,Y)$,  and $\VPV$-functions $r_1,r_2$, such that \[\V^1(\VPV)\vdash \left[\forall x,X \, \exists y,Y\, F(x,X,y,Y) \right]\supset \varphi(x,X,r_1(x,X,y,Y),r_2(x,X,y,Y))\footnote{In fact, they only need $\V^0$! This is shown by proving $\Iter$ with $\AC^0$ circuits is still complete for $\PLS$.}\] 
\end{theorem}

We additionally need one more claim for converting propositional unsatisfiable formulas into quantified tautologies.

\begin{claim}\label{lem:conv_to_exists}
 Let $F = \bigwedge_{i=1}^m C_i$ be an unsatisfiable CNF, and suppose $G_1 \vdash \lnot F$. Then \[G_1 \vdash \exists i \, \Big(\llbracket i = 1 \rrbracket \lor \cdots \lor \llbracket i = m \rrbracket\big) \land \big( (\llbracket i = 1 \rrbracket \rightarrow \lnot C_1) \land \cdots \land ( \llbracket i = m \rrbracket \rightarrow \lnot C_m)\Big).\]
    This is the quantified version of the $\unsat_F$ formula of Definition \Cref{def:unsat-cnf}, and it is near identical to the first-order formula $\unsat_F$ used to define $\mathsf{Reduction}_\iter$. We will denote this as $\unsat^q_F$.
\end{claim}

\begin{lemma}\label{lem:wit_total}
    $\mathsf{Wit}(F, \Pi, x)$ is provably total in $\TV^1$.
\end{lemma}
\begin{proof}
    We need the following well-known facts.
    \begin{enumerate}
        \item  For fixed $(F,\Pi)$ where $\Pi$ is a valid $G_1$-proof of $\Sigma^q_1$-formula $F$, $\V^1 \vdash \mathsf{Proof}_{G_1}(F, \Pi)$, 
        \item $\TV^1 \vdash \Sigma^q_1-\mathsf{Refl}_{G_1}(F,\Pi,x)$
        
    \end{enumerate}
See \cite{Krajicek2019,CN2010} for proofs of these facts. As well, see Appendix \ref{appendix:BA} for a brief discussion of $\Sigma^q_1$-$\refl_P$.

By definition of $\mathsf{Wit}(F,\Pi,x)$, if $\Pi$ is not a valid proof of $F$ then we output $\bot$. For $F,\Pi$ satisfying $\mathsf{Proof}_{G_1}(F,\Pi)$, we additionally have $\TV^1 \vdash \mathsf{Proof}_{G_1}(F,\Pi)$. Let $F \equiv \exists y \varphi(x,y)$, for $\varphi\in\Sigma^q_0$. As $\TV^1$ proves $G_1$ reflection, if $G_1 \vdash F$, then $\TV^1 \vdash \exists Y \, \varphi(X, Y)$. 

To complete the proof, we only need to give a $\V^1(\VPV)$ definition of $\mathsf{Wit}(F,\Pi)$. We need the following standard $\VPV$ predicates and functions\footnote{With more work, these are further defined in $\mathsf{VTC}^0$, \cite{CN2010}}:

\begin{itemize}
    \item $\mathsf{Fmla}(F)$ --- Expresses that $F$ is a valid quantified propositional formula,
    \item $\mathsf{Sub}(F,X,W)$ --- Outputs $F$ with the terms $X$ substituted into free variables, and $W$ substituted into existentially quantified variables, removing the quantifiers
\end{itemize}

 We may define $\mathsf{Wit}(F,\Pi,X) \equiv \mathsf{Fmla}(F) \supset \exists Y\, \mathsf{SAT}(\mathsf{Sub}(F,X,Y))$. Given this definition, it is clear that since $\TV^1 \vdash \exists Y \, \varphi(X, Y)$, then $\TV^1 \vdash \forall X  \mathsf{Wit}(F,\Pi,X)$. 
\end{proof}
\begin{lemma}\label{lem:search_to_wit}
     $\search_F$ reduces to $\mathsf{Wit}_{F, \Pi}$ in $\V^1$. That is, for $\Pi$ a $G_1$ refutation of CNF $F$, $\V_1 \vdash \mathsf{Reduction}_{\mathsf{Wit}_{F, \Pi}}$.
\end{lemma}
\begin{proof}
    The two polytime functions $A(F, X)$, $B(F, X, i)$ in the reductions are the obvious ones: $A(F, X)$ maps $F$ to $\text{UNSAT}^q_F$, with $X$ substituted in for the variables, and $B(F, X, i)$ maps index $i \in [m]$ to clause $C_m$ of $F$. $\V^1$ proves these functions form a valid reduction to $\mathsf{Wit}(F,\Pi)$ by proving $\forall i \leq m\, \mathsf{SAT}(\mathsf{Sub}(\text{Unsat}^q_F, X, i)) \supset \lnot \mathsf{Clause}(F,i)$. This follows by definition of $\text{UNSAT}^q_F$.
\end{proof}
We are now ready to prove the main theorem of this subsection, the equivalence between $G_1$ and $\langle EF,\Iter\rangle$.

\begin{proof}[Proof of \Cref{claim:g1_to_mapping}]
	We begin by proving $G_1 \geq_p \langle EF, \iter\rangle $. Suppose there is a $G_1$ refutation of an unsatisfiable CNF formula $F$. Then by \Cref{lem:conv_to_exists}, $G_1 \vdash \text{UNSAT}^q_F$. $\V^1$ proves $\Reduction_\iter(n,m,C,D,P,N)$ for circuits $C,D$ defined by combining the reduction of $\search_F$ to $\mathsf{Wit}(F,\Pi)$ (\Cref{lem:conv_to_exists}) with the provable totality of $\mathsf{Wit}(F,\Pi)$ in $\TV^1$ (\Cref{lem:wit_total}) and $\TV^1$ witnessing (\Cref{thm:eff_tv1_witnessing}). Applying the standard propositional translation for $\V^1$, we get that there is an $EF$ proof $\tau$ of $\|\mathsf{Reduction}_\iter\|$, with size $|\tau| = O(|\Pi|^k)$, for some fixed $k$ dependent only on the proof of \Cref{thm:eff_tv1_witnessing}. 

    Next, we prove the converse direction, that $G_1 \leq_p \langle EF, \iter\rangle$.  It is well known that if the theory $\mathsf{TV}^1$ proves the reflection principle of a propositional proof system $P$, then $P \geq_p G_1$ \cite{krajicek1995bounded}. Hence, we only need to show that $\mathsf{TV}^1$ proves the reflection principle of $\langle EF, \iter\rangle$. We give the explicit formula with a \emph{refutation} predicate $\mathsf{Ref}$ indicating a valid refutation, and a \emph{proof} predicate $\mathsf{Proof}$ indicating a valid proof. We need both to sensibly define reflection for $\langle EF, \iter \rangle$.

    \[\mathsf{Refl}_{\langle EF, \iter\rangle} \coloneq \| \mathsf{Ref}_{\langle EF, \iter \rangle}(F, C, D, \Pi) \supset  \lnot \mathsf{SAT}(F, X) \|,\]
    where $\mathsf{Ref}_{\langle EF, \iter \rangle}(F, C, D, \Pi) \coloneq \mathsf{Proof}_{EF}(\Reduction_\iter(n,m,C,D,P,N), \Pi)$ and $X$ is an assignment to the variables of $F$.
    
    Assume $\mathsf{Ref}_{\langle EF, \iter \rangle}(F, C, D, \Pi)$ is true. Then, since $\V^1/\TV^1 \vdash \refl_{EF}$, we have $$\TV^1 \vdash \mathsf{SAT}(\|\Reduction_\iter(n,m,C,D,P,N), X)\|).$$ 
    As a reminder, the formula $\Reduction_\iter(n, m, s, C, D, P, N)$ is defined to be the conjunction of the formulas:
    \begin{itemize}
        \item $\isCkt(n, C)$,
        \item $\isCkt(n+s, D)$,
        \item and the formula
        \begin{multline*}
            \forall X < n, Y < s: \iter(s, \Eval(n, s, C, X), Y) \rightarrow\\
            \unsat(n, m, P, N, X, \Eval(n+s, m, D, X \circ Y)).
        \end{multline*} 
    \end{itemize}
    Since $\TV^1$ proves the totality of $\Iter$, we can cut on the formula 
            $\forall X < n, Y < s: \iter(s, \Eval(n, s, C, X), Y) $  to derive $\unsat_F$, which is provably equivalent to $\lnot \mathsf{SAT}(F,X)$ in $\V^1$.
\end{proof}

\paragraph{A Comparison with $G^*_1$ and $\FP$-Witnessing.}

Wang \cite{Wang13g} showed that $[EF, \text{TreeRes}]$ is polynomially equivalent to Extended Frege. Our proof simplifies and generalizes their proof strategy, which we summarize below.

To show that $[EF, \text{TreeRes}] \geq_p EF$, the same strategy as ours is used, by setting the succinct tree-like Resolution proof to be the trivial exponential size proof, and using the original Extended Frege proof to show that the circuit is valid. For the reverse of $[EF, \text{TreeRes}] \leq_p EF$, Wang roughly does the following: 
\begin{enumerate}
    \item Let $C$ succinctly represent proof $\Pi$, and $\Pi_C$ be an $EF$ proof of the validity of $C$. Use extension variables on the evaluation of $C$ to define a path of polynomial length through the succinctly representing tree-like resolution proof $\Pi$, from root to axiom,
    \item As $EF$ proves the reflection principle of tree-like resolution, there must be a root-to-leaf path from $\bot$ to a falsified axiom of unsat CNF $F$. 
    \item Combine with $\Pi_C$ to get a complete refutation of $F$.
\end{enumerate}

This proof implicitly runs through $G^*_1/\V^1$-witnessing without any mention, as the algorithm used to witness a root-to-leaf path finding a falsified clause is going to be the same witnessing algorithm suggested in the $\V^1$-witnessing of \cite{CN2010}. Our proof of $[EF,Res] = G_1$ identifies that \emph{efficiently provable witnessing} is the lynchpin of these equivalences.

\subsection{Generalizing to \texorpdfstring{$G_i$}{Gi}}

By combining known witnessing theorems in the literature, as well as the connection between quantified proof systems $G_i$ and the theories $T^i_2$ \cite{krajicek1990quantified}, we can extend our main result to stronger systems.
Kraj\'i\v cek, Skelley, and Thapen \cite{krajivcek2007np}, as well as Skelley-Thapen \cite{skelley2007provably}, characterize the $\forall\Sigma^B_1$-consequences of $\TV^i$ by the \emph{game induction principles}, $\mathsf{GI}_i$, based on $i$-turn games.\footnote{Note that \cite{krajivcek2007np} do this for the first order theories $\T^i_2$, which are equivalent to $\TV^i$ by the RSUV isomorphism \cite{CN2010}.} In parallel work of Beckmann and Buss \cite{BeckmannB09}, they characterize the $\forall\Sigma^B_k$-consequences of $\TV^i$, $0 \leq k \leq i$ by proving a general witnessing theorem. 
Moreover, they show their witnessing theorem is efficiently provable, like \Cref{thm:eff_tv1_witnessing}, in $\V^1$. 
By combining these two works, we can generalize our main theorem to the following.

\begin{theorem}
    $G_i \equiv_p \langle EF, \mathsf{GI}_i\rangle$.
\end{theorem}

The proof would follow the same template as \Cref{sec:g1_imp_res}, and uses that game induction is provably reducible to the witnessing problems of Beckmann-Buss.

\section{A Generic Correspondence} \label{sec:generic}

In \cref{sec:G1} we showed an equivalence between $G_1$ and $\langle EF,\Res\rangle$. In this section we state necessary and sufficient conditions for a  proof system to be equivalent to EF-provable reductions to a $\TFNP$ problem. Our results in this section can be seen as generalizing 
the equivalence between black-box $\TFNP$ classes and proof systems from \cite{BussFI23} to the white-box setting, under Extended Frege-provable reductions. At a high level, they showed that any $\TFNP^{dt}$ class (satisfying certain mild conditions) can be viewed as the total $\NP$ search problems that are decision tree reductions to the \emph{reflection principle} (encoded as a propositional formula) for some proof system. For example, the problems reducible to $\Iter^{dt}$ are the problems reducible to the reflection principle for narrow Resolution. 
Conversely, they showed that if a proof system proves its own reflection principle and can simulate decision tree reductions then it is characterized by a $\TFNP^{dt}$ class.  

In this section, we show that if a proof system  $P$ can prove its own reflection principle (cf. \Cref{def:reflection}) and is at least as strong as Extended Frege, then EF-provable reductions to the \emph{Wrong Proof problem} for $P$ are equivalent to $P$ proofs. 
\begin{definition} $\WrongProof_P:= \{\WrongProof_{P,n,m,s}\}_{n,m,s}$ is the false clause search problem associated with the unsatisfiable formula $\neg \refl_P$, which asserts that a formula has both a $P$-proof and a falsifying assignment. Formally, \begin{align}
    \WrongProof_{P,n,m,s}(x,z) :&= \unsat_{\neg \refl_{P,n,m,s}}(x,z) \nonumber \\
    &=\unsat_{\Proof_{P,n,m,s}}(x,z_1) \lor \unsat_{\neg \mathrm{SAT}_{n,m}}(x,z_2),
\end{align}
Recalling that we use a one-hot encoding in $\unsat$, $z_1$ are the bits that represent the clauses of $\Proof$, while $z_2$ are those belonging to $\mathrm{SAT}$. As well, $x=(\alpha,F,\Pi)$ is a tuple encoding an assignment, a formula (claimed to be a tautology), and a proof $\Pi$.
Recall from \Cref{def:unsat-cnf} that this is a problem in $\TFNP$. 
\end{definition}

Informally, $\WrongProof_P$ is the following search problem:
Given a CNF formula $F$, an assignment $x$ to the variables of $F$, and a proposed $P$-refutation $\Pi$ of $F$, either output a clause of $F$ that is falsified under $x$ or an error in $\Pi$.
We note that the $\WrongProof$ problem was originally introduced by Goldberg and Papadimitriou \cite{GoldbergP18} in a slightly different setting --- our usage is closer to that of \cite{li2024metamathematics}.
The main theorem of this section says that if a sufficiently strong proof system $P$ can prove its own reflection principle, then it is characterized by $EF$-provable reductions to $\WrongProof_P$.

\begin{theorem}
\label{thm:generalEquiv}
    Let $P \geq_p EF$ be any proof system which has polynomial-size proofs of $\refl_{P}$. Then $P$ is polynomially equivalent to $\langle EF,\WrongProof_P\rangle$
\end{theorem}
Said differently, there is a size-$s$ $P$-refutation of $F$ iff there is an $EF$-provable many-one reduction from $\search_F$ to $\WrongProof_P$ where both the reduction and the Extended Frege proof have size polynomial in $s$. 
It may be helpful to keep in mind the proofs of \Cref{prop:EF_simulates_FP} and \Cref{prop:FP_simulates_EF}, as our proof of this theorem follows a similar trajectory. 

 We will use the following simple properties of Extended Frege.

\begin{claim}
\label{clm:EF_properties}
    For any CNF formula $F$ on $n$ variables, the following hold:
    \begin{enumerate}
        \item Extended Frege has $\poly(|F|)$-size proofs of $\neg  F(x) \rightarrow \unsat_F(x,y)$, for new variables $y$, and $\unsat_F(x,y) \rightarrow \neg F(x)$.
        \item For any circuit $C: \{0,1\}^t \rightarrow \{0,1\}^n$, if Extended Frege can prove $F$ then it also has $\poly(|C|)$-size proofs of  $\Eval(C,x',x) \wedge F(x)$.
    \end{enumerate}
\end{claim}
\begin{proof} 
    We sketch the proof. Let $F=C_1 \wedge \ldots \wedge C_m$.
    For the first item, Extended Frege begins by deriving $F \lor \neg F$. Then, it introduces new variables $y$ such that $y_i \iff \neg C_i$. From $\neg F = \neg C_1 \lor \ldots \lor \neg C_m$ and the definition of the $z_i$'s it then derives $\unsat_F(x,z)$ to obtain $F \rightarrow \unsat_F(x,z)$.
    The proof of the converse is similar. 

    The second item is standard. Beginning with $F(x)$, and treating $x$ as the output gates of $C$, Extended Frege introduces new variables $x'_1,\ldots, x'_m, g_1,\ldots, g_{|S|}$, where $g_i$ will represent the value of the $i$-th gate of $C$, such that the input variables are $x'$, and clauses enforcing that the value of $g_i$ is correctly derived from the values of its children in $C$. 
\end{proof}

We prove \Cref{thm:generalEquiv} over the new two lemmas, the forward direction is captured by \Cref{lem:forwardGeneral} taking $P=Q$, while the backwards direction is given in \Cref{lem:backwards}.

\begin{lemma}
\label{lem:forwardGeneral}
    Let $P,Q$ be proof systems such that $EF \leq_p P \leq_p Q$
    and $Q$ proves the reflection principle for $P$. If $\langle EF,\WrongProof_{P}\rangle$ has a size $s$ refutation of $F$, then $Q$ has a size $\poly(s)$ refutation of $F$. 
\end{lemma}
\begin{proof}
    Let $(C,D,\Pi)$ be an $\langle EF, \WrongProof_P \rangle$ refutation of an unsatisfiable CNF formula $F$. In particular, $(C,D)$ is a mapping reduction and $\Pi$ is an Extended Frege proof of $\Reduction_{\WrongProof_P}$. We will derive the premise of $\Reduction_{\WrongProof_P}$,
        \begin{align} \label{eq:premise}
        \isCkt(C) \wedge \isCkt(D) \wedge \Eval(C,x,x') \wedge \Eval(D,(x,y),z) \wedge \WrongProof_P(x',y),\end{align}
    and then cut it with the conclusion of $\Pi$ to prove $\neg F$. 

    First, we use that $P$ has small proofs of $\refl_P$ in order to derive $\unsat_{\neg \refl_P}(x',y)$, which is equivalent to $\WrongProof_P(x',y)$.
    More formally, by \autoref{clm:EF_properties}, Extended Frege can prove \[\neg \refl_P(x') \implies \unsat_{\neg \refl_P}(x',y), \]
    which follows by first writing down the statement $\neg \refl_P \lor \refl_P$ and then, from $\neg \refl_P$, introducing new variables $y$ and deriving  $\unsat_{\neg \refl_P}(x',y)$. Cutting this formula with $\refl_P(x')$, we obtain $\WrongProof_P(x',y)$.

    Finally, we derive the statements about the circuits $C$ and $D$. Since these are fixed circuits, $\isCkt(C)$ and $\isCkt(D)$ are constant true formulas. As well, by \autoref{clm:EF_properties}, Extended Frege can derive $\Eval(C,x,x') \wedge \Eval(D,(x,y),z)$, where $C$ and $D$ are fixed. This completes the derivation of the premise (\ref{eq:premise}). Using the supplied proof $\Pi$ we can derive $\Reduction_{\WrongProof_P}$. Cutting this with (\ref{eq:premise}) gives $\unsat_F(x,z)$, from which we can deduce $\neg F$ using \Cref{clm:EF_properties}.
\end{proof}

\begin{lemma}
\label{lem:backwards}
    Let $P \geq_p EF$ be any proof system such that $P$ has polynomial-size proofs of $\refl_P$. If there is a size-$s$ $\langle EF, \WrongProof_P \rangle$ of $F$ then there is size $\poly(s)$ proof of $F$.
\end{lemma}

\begin{proof}[Proof of \Cref{lem:backwards}]
 For the backwards direction, let $\Pi'$ be a $P$-refutation of $F$, and let $\Pi$ be a similar-sized proof that $\neg F$ is a tautology (for example by cutting $\Pi$ with $F \lor \neg F$. We will need $\Pi$ as we have phrased $\refl$ for proofs of tautologies. 
 
 We construct a mapping reduction $(C,D)$ from $\search_F$ to $\WrongProof_{P}$. Recall that $\WrongProof_P$ has three inputs: a formula, an $n$-bit truth assignment claiming to falsify the formula, and a $P$-proof claiming that the formula is a tautology. On input $x \in \{0,1\}^n$, the circuit $C$ maps $x$ to the truth assignment variables, and hard-codes $\Pi$ for the proof variables and $\neg F$ for the formula variables. By assumption, $\Pi$ is a valid $P$-proof of $\neg F$ and so $\Proof(C(x)) = \Proof(\Pi,\neg F)$ is the constant true formula. As well, denoting by $\mathrm{SAT}_{\neg F}$ the formula $\mathrm{SAT}$ with its formula-input fixed to $\neg F$, observe that 
    \[\unsat_{\neg \mathrm{SAT}}(A(x),z) = \unsat_{\neg \mathrm{SAT}_{\neg F}}(x,z) =\unsat_F(x,z), \] by definition.
    Hence,
    \begin{align*} 
        \WrongProof_{P}(C(x),z) &= \unsat_{\Proof}(C(x),z_1) \lor \unsat_{\neg \mathrm{SAT}}(C(x),z_2) \\
        &=\bot \lor \unsat_{\neg \mathrm{SAT}_{\neg F}}(x,z_2) \\
        &= \unsat_F(x,z_2).
    \end{align*}
    Finally, define the circuit $D$ to map $(z_1,z_2) \mapsto z_2$, recalling that we use a one-hot encoding for $z$.

    It remains to construct an Extended Frege proof of $\Reduction_{\WrongProof_P(C,D,F)}$ using the assumed proof $\Pi$. Beginning from the premise of this statement,
    \[ \Eval(C, x, x') \land \Eval(D, (x, y), z) \land \WrongProof_P(x', y),\]
    where we have used that, since $A$ and $D$ are fixed, $\isCkt(C)$ and $\isCkt(D)$ are constant true formulas.
    Extended Frege must argue two things:
    \begin{enumerate}
        \item That $C$ hard-wires $(\neg F, \Pi)$ and so $\Proof(\Pi,\neg F)$ is the constant true formula, and that $D$ maps indices of falsified clauses of $\neg \mathrm{SAT}(x,\neg F)$ to indices of falsified clauses of $F$.
        \item That $\Pi$ is an $P$-proof of $F$. 
    \end{enumerate}
    Examining the gates of the circuit $C$, Extended Frege can deduce that the output string of $\Eval(C,x,x')$ has the form $x'=(x,\neg F,\Pi)$. Hence, it can also derive that 
    \[\WrongProof_P(x',y) = \unsat_{ \Proof}(\neg F,\Pi,y_1) \lor \unsat_{\neg \mathrm{SAT}}(x,\neg F,y_2),\]
    where $\neg F,C,\Pi$ are fixed, $y= y_1 \circ y_2$, and we have expanded the definition of $\WrongProof$. Use $\Pi$ to derive $\Proof(\neg F,\Pi)$ and cut it on  $\unsat_{\neg \Proof}(F,\Pi,y_1)$ in order to derive 
    \[ \unsat_{\mathrm{SAT}}(x,\neg F,y_2). \]
    Finally, examining the gates of $D$, Extended Frege can prove that the output string of $\Eval(D,(x,y),z)$ satisfies $z=y_2$, where $y= y_1 \circ y_2$. Therefore, it can derive 
    \[ \unsat_{\neg \mathrm{SAT}}(x,\neg F,z) = \unsat_F(x,z),\]
    where equality holds syntactically by definition since $F$ is fixed.
\end{proof}

\addcontentsline{toc}{section}{References}
\bibliographystyle{alpha}
\bibliography{implicit,tfnp-refs}

\appendix

\section{Appendix: Proof Systems and Bounded Arithmetic}
\label{sec:appendix}
\subsection{Propositional and QBF Proof Systems}
We formally describe proof systems Frege, Extended Frege, and $G_i$. Refer to \cite{Krajicek2019} for results on the strength of these systems.

\paragraph{Frege and Extended Frege.}

\begin{definition}[Gentzen's PK: Propositional Calculus]\label{def:pk}
Let all formulas be propositional and under the DeMorgan basis. A \emph{PK proof} uses the following rules. 

\[
\begin{array}{l l}

\text{(Ax)}\;
\inferrule*[right=Ax]{ }{A \rightarrow A}
&
\text{(Cut)}\;
\inferrule*[right=Cut]{\Gamma \rightarrow A \\ \Delta, A \rightarrow B}{\Gamma, \Delta \rightarrow B}
\\

\text{(WL)}\;
\inferrule*[right=WL]{\Gamma \rightarrow A}{\Gamma, B \rightarrow A}
&
\text{(WR)}\;
\inferrule*[right=WR]{\Gamma \rightarrow A}{\Gamma \rightarrow A, B}
\\

\text{(CL)}\;
\inferrule*[right=CL]{\Gamma, A, A \rightarrow B}{\Gamma, A \rightarrow B}
&
\text{(CR)}\;
\inferrule*[right=CR]{\Gamma \rightarrow A, A}{\Gamma \rightarrow A}
\\

\text{(EL)}\;
\inferrule*[right=EL]{\Gamma, A, B, \Delta \rightarrow C}{\Gamma, B, A, \Delta \rightarrow C}
&
\text{(ER)}\;
\inferrule*[right=ER]{\Gamma \rightarrow A, B, \Delta}{\Gamma \rightarrow B, A, \Delta}
\\

\text{(neg L)}\;
\inferrule*[right=neg L]{\Gamma \rightarrow A, \Delta}{\Gamma, \neg A \rightarrow \Delta}
&
\text{(neg R)}\;
\inferrule*[right=neg R]{\Gamma, A \rightarrow \Delta}{\Gamma \rightarrow \neg A, \Delta}
\\

\text{(and L)}\;
\inferrule*[right=and L]{\Gamma, A, B \rightarrow \Delta}{\Gamma, A \wedge B \rightarrow \Delta}
&
\text{(and R)}\;
\inferrule*[right=and R]{\Gamma \rightarrow A, \Delta \\ \Gamma \rightarrow B, \Delta}{\Gamma \rightarrow A \wedge B, \Delta}
\\

\text{(or L)}\;
\inferrule*[right=or L]{\Gamma, A \rightarrow \Delta \\ \Gamma, B \rightarrow \Delta}{\Gamma, A \vee B \rightarrow \Delta}
&
\text{(or R)}\;
\inferrule*[right=or R]{\Gamma \rightarrow A, B, \Delta}{\Gamma \rightarrow A \vee B, \Delta}

\end{array}
\]
\end{definition}
\noindent
Extended Frege, or $ePK$, allows the addition of \emph{extension variables} to $PK$.

\begin{definition}
    Let $F$ be a $PK$-derivation from initial formulas $F_1,\dots,F_k$ over variables $p_i$. Let $e$ be a variable not among the $p_i$ variables of $F$. An \emph{extension rule} is the inference, $\rightarrow e \leftrightarrow A$, for $A$ a formula over variables $p_i$. We say that $e$ is an \emph{extension variable}.
\end{definition}
\noindent
Extended Frege, or $ePK$, is the propositional proof system of $PK$, with the addition of extension rules. 

\paragraph{Quantified Propositional Calculus and $G_i$ systems.}
Adding quantified formulas and rules for quantifiers defines the more powerful $G$ proof system and $G_i$ subsystems.

The proof system $G$ generalizes Frege and Extended Frege systems to general quantified Boolean formulas (QBF). 

\begin{definition}[$G$ and $G_i$ Subsystems]
\label{def:G}
	We adapt Frege to allow quantified formulas introduced by the following rules:

	\[
\begin{array}{l l}

(\forall\text{-left})\;
\inferrule*{A(B),\, \Gamma \rightarrow \Delta }{\forall x A(x), \Gamma \rightarrow \Delta}
&
(\forall\text{-right})\;
\inferrule*{\Gamma \rightarrow \Delta, A(p)}{\Gamma \rightarrow \Delta, \,\forall x\, A(x)}
\\

(\exists\text{-left})\;
\inferrule*{A(p), \,\Gamma \rightarrow \Delta }{\exists x A(x),\, \Gamma \rightarrow \Delta}
&
(\exists\text{-right})\;
\inferrule*{\Gamma \rightarrow \Delta, A(p)}{\Gamma \rightarrow \Delta, \,\exists x\, A(x)}

\end{array}
\]

For the $\forall$-right and $\exists$-left rules, the variable $p$ is called an \emph{eigenvariable} and cannot appear in the bottom sequent. The $G$ proof system is simply $PK$ with these additional rules. For $i\geq 0$, $G_i$ is a restriction of $G$ to only allow cuts on $\left(\Sigma^q_i \cup \Pi^q_i\right)$-formulas. (That is, the cut rule for $G_i$ can only be applied to formulas with at most $i$ alternations of $\forall$ and $\exists$ quantifiers.) 

While the $G$ systems can prove validity of  general QBF formulas, it is also interesting to understand their strength with respect to Boolean formulas.
Since $G$ (and even $G_1$) has all of EF rules plus additional rules, $G$ can trivially p-simulate EF (with respect to Boolean formulas). 
And it is known that the tree-like version of $G_1$, called $G_1^*$ is p-equivalent to EF. On the other hand, it is unknown whether or not $G_1$ is  stronger than EF (with respect to Boolean formulas), and it is generally conjectured to be more powerful wrt p-simulations.

\end{definition}

\subsection{Bounded Arithmetic}
\label{appendix:BA}

We will assume certain preliminaries from bounded arithmetic, although we review some of the main points now. 
We work in \emph{second-order} bounded arithmetic, where we are reasoning about objects with two types: \emph{natural numbers} (denoted with lower-case variables $x, y, z$) and \emph{sets of numbers} (denoted with upper-case variables $X, Y, Z$).
We use the standard language \[\mathcal{L}^2_A := (0,1,+,\cdot, |\cdot|; =_1, =_2, \leq, \in)\] where $0, 1, +, \cdot$ are number terms representing $0, 1$ and usual addition and multiplication, and \[|S| = \begin{cases} 0 & S = \emptyset \\ 1 + \displaystyle \max_{i \in S} i  & S \neq \emptyset\end{cases}\] is a function which takes a set term and outputs a number which is one larger than the largest element of the set.
The relations $=_1$ and $=_2$ are equality symbols for numbers and strings, respectively, $\leq$ is the normal ordering on numbers, and $x \in X$ is intended to mean usual set inclusion. 
We write $X(x) := x \in X$.

It is convenient to think of a set of numbers as a binary string. 
We follow Cook-Nguyen (CITE) and associate the set $S$ with the string $w(S) = S(n)S(n-1)\cdots S(0)$, where $n$ is the largest element of the set $S$. 
Note that $|S|$ is the length of the string associated with $S$.
Moving forward we will refer to set terms as string terms, and think almost exclusively in terms of strings.

We sometimes treat strings as multi-dimensional string arrays, which can be done by standard methods using pairing relations. 
In this way, if $X$ is an $n \times m$ string array, we use indexing notation $X(i, j)$ to denote the corresponding element at the index $(i, j)$. 
We use $X^{[i]}$ to denote the length-$m$ string encoded in the $i$th row of $X$.
We will also sometimes need to take array slices. 
For instance, if $|X| = n$ and $i < j \leq n$, then we write $X(i:j)$ to mean the substring that starts at index $i$ and ends at index $j$, including the first endpoint of the interval but excluding the last endpoint.
If $j \leq i$ then $X(i:j)$ is just the empty string.
We can also combine these notations, writing e.g.~$X(i, j:k)$.
All of these relations are definable using $\Sigma^B_0$ string comprehension and hence can be used freely. 

\paragraph{$V^1$ and $TV^1$. } 

We follow \cite{CN2010} in the presentation of $V^1$, a two-sorted theory for polynomial time which is equivalent to Buss' first order theory, $S^1_2$. 

At the base of $V^1$ are the 2-BASIC axioms, which define the basic behaviors of the function symbols and predicates of $\mathcal{L}^2_A$. See \cite{CN2010} for a full list.

\begin{definition}[$\Sigma^B_i$-Comprehension]
	For a formula $\varphi(x)$, the $\varphi$-COMP axiom is the following sentence, \[\exists |X| \leq y \, \forall i < y . \, X(i) = 1 \longleftrightarrow \varphi(i).\]
	If $\Gamma$ is a set of formulas then $\Gamma\text{-COMP} := \set{\phi\text{-COMP} : \phi \in \Gamma}$.
	We will be particularly interested in the collections $\Sigma^B_i$-COMP for each $i$.
\end{definition}

The theory $V^1$ is then defined as,
\[V^1 \coloneq \text{2-BASIC} + \Sigma^B_1\text{-COMP}.\]

\begin{theorem}[$\FP$ Definability, \cite{CN2010}]
	A function $f$ is in $FP$ if and only if $f$ is $\Sigma^B_1$-definable in $V^1$. 
\end{theorem}

We list other well-known schemas which $V^1$ is able to prove.

\begin{lemma}[\cite{CN2010}] (i) $V^1 \vdash \Sigma^B_1$-IND, (ii) $V^1 \vdash \Sigma^B_1$-MIN, and (iii) $V^1 \vdash \Sigma^B_1$-MAX, referring to the standard induction, minimization, and maximization principles respectively. 
\end{lemma}

The theory $V^1(VPV)$ is obtained from $V^1$ by extending the language $\mathcal{L}^2_A$ with $\mathcal{L}_{\FP}$, which has a function symbol for every polynomial-time function. 
The axioms of $V^1(VPV)$ are the union of the axioms of $V^1$ with the axioms of $VPV$, which contain defining Cobham axioms for each of the symbols in $\mathcal{L}_{\FP}$.
Cook and Nguyen prove that $V^1(VPV)$ is a conservative extension of $V^1$, and it can prove the corresponding comprehension and induction axioms for $\Sigma^B_1(\mathcal{L}_{\FP})$ formulas. 

The theory $\mathsf{TV}^1$ is a stronger theory than $\mathsf{V}^1$, axiomatized by the \emph{string induction} schema. Interpreting a string $X$ as a binary number, we may $\Sigma^B_0$ define a string successor function $S(X)$ that adds $1$ to $X$ via binary addition. See \cite{CN2010} for the definition.

\begin{definition}
    String induction is the two-sorted version of full induction. Let $\varphi$ be a formula. \[\varphi\text{-}\mathsf{SIND} \triangleq \varphi(\epsilon) \land \left( \forall X\, \varphi(X) \rightarrow \varphi(S(X)) \right) \rightarrow \forall \, X\, \varphi(X).\]
\end{definition}

With string induction, we may define $\mathsf{TV}^1$:

    \[\mathsf{TV}^1 \equiv \text{2-BASIC} + \Sigma^B_1\text{-SIND}.\]

\begin{theorem}[\cite{CN2010}]
    $\mathsf{TV}^1 \supset \mathsf{V}^1$.
\end{theorem} 

As we will see below, $\mathsf{TV}^1$ corresponds to ``$\PLS$" reasoning in regards to the provability of $\forall\Sigma^B_1$-sentences. 

\subsection{Witnessing Theorems}

We recall several witnessing theorems.

\begin{theorem}[See \cite{CN2010}]\label{thm:v1_witnessing}
	Let $\phi(\vec x, \vec y, \vec X, \vec Y)$ be a $\Sigma^B_0$ formula.
	If \[V^1 \vdash \forall \vec x, \vec X \exists \vec y, \vec Y \phi(\vec x, \vec y, \vec X, \vec Y),\] then there are polynomial-time functions $f_1, \dots, f_k$, $F_1, \dots, F_m$ such that \[V^1(\vec f, \vec F) \vdash \forall \vec x, \vec X \phi(\vec x, \vec f(\vec x, \vec X), \vec X, \vec F(\vec x, \vec X)).\]
\end{theorem}

We also have the following related which holds for $V^1(VPV)$ \cite{CN2010}.

\begin{theorem}
	Let $\phi(\vec x, \vec y, \vec X, \vec Y)$ be a $\Sigma^B_0(VPV)$ formula.
	If \[V^1(VPV) \vdash \forall \vec x, \vec X \exists \vec y, \vec Y \phi(\vec x, \vec y, \vec X, \vec Y),\] then there are polynomial-time functions $f_1, \dots, f_k$, $F_1, \dots, F_m$ such that \[V^1(VPV) \vdash \forall \vec x, \vec X \phi(\vec x, \vec f(\vec x, \vec X), \vec X, \vec F(\vec x, \vec X)).\]
\end{theorem}

For $\mathsf{TV}^1$, we get $\PLS$-witnessing.

\begin{theorem}[\cite{CN2010}]\label{thm:tv1_witnessing}
    Let $\phi(\vec x, \vec y, \vec X, \vec Y)$ be a $\Sigma^B_0(VPV)$ formula.
	If \[\mathsf{TV}^1 \vdash \forall \vec x, \vec X \exists \vec y, \vec Y \phi(\vec x, \vec y, \vec X, \vec Y),\] then there is a polynomial-time function $F$, and an $\Iter$ instance with graph $G_\phi(x, X, Z)$ such that \[\mathsf{V}^1(\mathsf{VPV})  \vdash  G_\phi(\vec x, \vec X) \longrightarrow \phi(\vec x, \vec X, F(\vec x, \vec X, Z)).\]
\end{theorem}

As $\mathsf{TV}^1$ proves the totality of $\Iter$, we can derive a statement more like Theorem \ref{thm:v1_witnessing}. Notice, however, that one only needs $\mathsf{V}^1$ to ``prove" the witnessing theorem. We will need this finegrained aspect of Theorem \ref{thm:tv1_witnessing} in Section \ref{sec:G1}. 
\subsection{Propositional Translations}

The following translations will allow us to relate second-order bounded arithemtic theories and propositional proofs. We follow the treatment in Cook-Nguyen \cite{CN2010}.

Let $\phi(\vec x, \vec X)$ be a $\Sigma^B_0$-formula with free number variables $\vec x$ and string variables $\vec X$. 
We define a polynomial-size, bounded-depth family of propositional formulas \[||\phi(\vec x, \vec X)|| = \set{||\phi(\vec x, \vec X)||_{\vec n, \vec m} : m_i, n_i \in \NN},\] where each $||\phi(\vec x, \vec X)||_{\vec m, \vec n}$ 
is a propositional formula defined such that it is valid if and only if for every $\vec x$ and $\vec X$, if $|X_k| = n_k$ for each $k$, then $\phi(\vec m, \vec X)$ is true in the standard model.

Let us now formally define propositional translations.
We use $0$ and $1$ for the boolean constants False and True, and write $\underline n$ to mean the numeral $n$. If $t$ is a closed term then $t^{\NN}$ is the value that $t$ takes in the standard model, and if $\psi$ is a closed sentence, we write $\psi^{\NN}$ to mean the truth value that $\psi$ takes under the standard model.
We begin by considering the case with a single string variable $\phi(X)$, and note that if $\phi$ had additional number variables $\phi(\vec x, X)$, then we define \[||\phi(\vec x, X)||_{\vec n, m} := ||\phi(\underline{\vec n}, X)||_m,\] and hence we can safely reduce to this case.

The propositional formula $||\phi(X)||_{m}$ is defined as follows by induction. 
Introduce $m$ propositional variables $p[X]_0, p[X]_1, \dots, p[X]_{m-1}$ representing the values of $X(i)$ for $i < m$.

\paragraph{Atomic Formulas.} 
For atomic formulas, we have the following cases:
\begin{itemize}
	\item If $\phi(X) := X = X$ then $||\phi(X)||_{m} = 1$.
	\item If $\phi(X) := t(|X|) = u(|X|)$ then $||\phi(X)||_{m} := (t(\underline m) = u(\underline m))^{\NN}$
	\item If $\phi(X) := X(t(|X|))$, then define \[||\phi(X)||_{m} := \begin{cases} 0 & m = 0 \\ p[X]_{t(m)^{\NN}} & t(m)^\NN < m-1 \\ 1 & t(m)^\NN = m-1 \\ 0 & t(m)^\NN > m-1 \end{cases}\]
\end{itemize}

\paragraph{Inductive Formulas.}
For the induction step, we define:
\begin{itemize}
	\item If $\phi(X) = \psi(X) \circ \nu(X)$ then $||\phi(X)||_{m} = ||\psi(X)||_{m} \circ ||\nu(X)||_{m}$, for $\circ \in \set{\land, \lor}$.
	\item If $\phi(X) = \neg \psi(X)$ then $||\phi(X)||_{m} = \neg ||\psi(X)||_{m}$.
	\item If $\phi(X) = \exists x \leq t(|X|) \psi(x, X)$ then $||\phi(X)||_{m} = \bigvee_{i=0}^{t(\underline m)^\NN} ||\psi(\underline i, X)||_m$
	\item If $\phi(X) = \forall x \leq t(|X|) \psi(x, X)$ then $||\phi(X)||_{m} = \bigwedge_{i=0}^{t(\underline m)^\NN} ||\psi(\underline i, X)||_{m}$.
\end{itemize}

The case of strictly bounded number quantifiers $\exists/\forall x < t$ is handled similarly. 
If we have a strictly bounded quantifier like $\exists x < 0$ or $\forall x < 0$, then the corresponding $\vee$ or $\wedge$ is empty, and hence we output $0$ and $1$, respectively.

If we have multiple string variables $\vec X$, then the above construction for $||\phi(\vec X)||_{\vec m}$ proceeds identically by introducing propositional variables $p[X_i]_j$ for each $X_i$ in the list $\vec X$ and each $j \leq m_i$.
Now we have one extra atomic formula case to handle, namely $\phi(\vec X) := X_i = X_j$ for $i \neq j$.
In this case, we can reduce to the single-variable case by defining \[||\phi(\vec X)||_{\vec m} := |||X_i| = |X_j| \land \forall x < |X_i|: X_i(x) \leftrightarrow X_j(x)||_{\vec m}.\]

Not only can first-order formulas be propositionally translated, but $\mathsf{V}^1$ proofs (of $\Pi^B_1$ formulas) can also be translated into polynomial-size Extended Frege proofs: 

\begin{theorem}[\cite{CN2010}]\label{thm:prop_trans_ef}
	Let $\varphi$ be $\Pi^B_1$ formula. If $\mathsf{V}^1 \vdash \varphi$, then $EF\vdash || \varphi||_n$. Furthermore, the $EF$ proofs are of polynomial size. 
\end{theorem}

This extends to the theory $\TV^1$ as well.

\begin{theorem}[\cite{CN2010}]\label{thm:prop_trans_g1}
    Let $\varphi$ be a $\Pi^B_1$ formula. If $\TV^1 \vdash \varphi$, then $G_1\vdash || \varphi||_n$. Furthermore, the $G_1$ proofs are of polynomial size. 
\end{theorem}
\subsection{Formalizing Propositional Logic in \texorpdfstring{$V^1$}{V1}}\label{sec:prop_logic}

We give an informal and high-level survey of the formalization of the syntax and semantics of propositional logic in bounded arithmetic. For a complete treatment, see \cite{CN2010}. The goal of this section will be to show a back-and-forth relationship between propositional proof systems and several standard theories of bounded arithmetic.

\paragraph{Syntax of Propositional Logic.} In $\V^1$ (and in fact, in much weaker theories), valid formula and $G_i$-proof predicates $\mathsf{Fmla}$, $\mathsf{Proof}_{G_i}$ are $\Delta^B_1$-definable. All that is required to design such a predicate are $\TC^0$-functions needed to verify proper syntax, as well as process parse trees of formulas.

\begin{definition}[Proofs, Satisfiability, and Reflection Principles]
    For a proof system $P$ (a polynomial-time verifier) and parameters $n,m,s$, a reflection principle for $P$ is a CNF formula 
    \[ \text{Refl}_P:=  \text{Proof}_P(F,\Pi, s) \supset \text{SAT}(F,X), \]
    where $\text{SAT}$ is the $\VPV$-predicate stating that $F$ is a formula satisfied by the truth assignment $X$. As well, $\text{Proof}_P$ is a $\VPV$-predicate which is true iff $\Pi$ is a $P$-proof of $F$. Such a formula can be obtained, for example, by taking the Cook-Reckhow translation of the verifying Turing Machine for $P$-proofs. Reflection principles establish the \emph{soundness} of a proof system.

    The satisfiability predicate and reflection principles may be generalized to quantified boolean formulas in the natural way, by changing the $\text{SAT}$ predicate to be $\text{SAT}_{\Sigma^q_1} \equiv \exists |W| < |F|\, \text{SAT}(F,X,W)$, where values for $W$ are substituted into existentially quantified variables of $F$. We refer to this reflection as $\Sigma^q_1$-reflection.
\end{definition}

\noindent Below are important theorems on reflection for theories $\V^1$ and $\TV^1$.

\begin{theorem}[Chapter X, \cite{CN2010}]
    $\V^1$ proves the reflection principle for Extended Frege, and the reflection principle (for $\Sigma^q_1$-formulas) for tree-like $G_1$. $\TV^1$ proves the reflection principle of $G_1$.
\end{theorem}

\begin{theorem}[Informal, \cite{CN2010}]\label{thm:back}
    Suppose that a theory $\T$, which contains $\V^1$, proves the reflection principle for propositional proof system $\mathcal{P}$. Then for a formula $F$, $P\vdash F \rightarrow \T \vdash SAT(F,X)$.
\end{theorem}
\begin{corollary}
    If EF $\vdash F$, then $\V^1 \vdash \text{SAT}(F,X)$. Similarly for $\TV^1$, if $G_1\vdash F$, then $\TV^1 \vdash \text{SAT}(F,X)$.
\end{corollary}

Theorem \ref{thm:back} establishes what is called a ``back-and-forth" relationship between bounded arithmetic theories and propositional proof systems. On the one hand, if $\V^1$ proves a $\forall\Sigma^B_1$-sentence $\forall X \exists Y < t.\,\varphi(X,Y)$, then by known propositional translations (Theorem \ref{thm:prop_trans_ef}) , $EF \vdash \|\forall X\exists Y<t.\, \varphi(X,Y)\|_n$. Moreover, this translation is \textit{provable} in $\V^1$. 

One the other hand, by Theorem \ref{thm:back}, if Extended Frege proves formula $F$, then $\V^1$ proves $F$ is a tautology. These two directions establish that we may go back-and-forth between the first-order setting and propositional proof systems for any pair $(\mathsf{T}, P)$ with the above relationships. In particular, we will need this between $\TV^1$ and $G_1$, as well as $\V^1(\VPV)$ and Extended Frege. For a complete picture of this relationship, see [Chapter X, \cite{CN2010}].

\end{document}